\documentclass{article}[11pt]
\usepackage[utf8]{inputenc}
\usepackage{amsmath}
\usepackage{amssymb}
\usepackage{cite}
\usepackage{makeidx}
\usepackage{graphicx}
\usepackage{multirow}
\usepackage{booktabs}
\usepackage{appendix}
\usepackage{enumitem}
\usepackage{bm}
\usepackage[normalem]{ulem}
\usepackage{subfig}
\usepackage{diagbox}
\usepackage{algorithm}
\usepackage{xcolor}
\usepackage{bigstrut}
\usepackage{algpseudocode}
\usepackage{amsfonts}
\usepackage{mathrsfs}

\newtheorem{theorem}{Theorem}
\newtheorem{lemma}[theorem]{Lemma}
\newtheorem{proposition}[theorem]{Proposition}
\newtheorem{definition}{Definition}
\newtheorem{corollary}{Corollary}
\newtheorem{proof}{Proof}

\numberwithin{equation}{section}
\numberwithin{theorem}{section}

\newcommand\bxi{\boldsymbol{\xi}}
\newcommand\bx{\boldsymbol{x}}

\newcommand\bu{\boldsymbol{u}}
\newcommand\bv{\boldsymbol{v}}

\newcommand\bA{\boldsymbol{A}}

\newcommand\bbA{\mathbb{A}}

\newcommand\bbR{\mathbb{R}}

\newcommand\bbV{\mathbb{V}}
\newcommand\bbW{\mathbb{W}}
\newcommand\bbF{\mathbb{F}}

\newcommand{\feq}{f_{\rm eq}}

\newcommand{\vp}{\varphi}

\newcommand\rd{\mathrm{d}}

\renewcommand{\lg}{\langle}
\newcommand{\rg}{\rangle}
\newcommand\Kn{\mathit{Kn}}
\newcommand\rspan{\operatorname{span}}

\newcommand\mL{\mathcal{L}}
\newcommand\mP{\mathcal{P}}
\newcommand\mQ{\mathcal{Q}}
\newcommand\mR{\mathcal{R}}

\newcommand{\wmQs}{\widetilde{\mQ}^{\ast}}
\newcommand\mbQ{\boldsymbol{\mathcal{Q}}}
\newcommand\mQs{\mathcal{Q}^{\ast}}
\newcommand\hmQ{\hat{\mQ}}
\newcommand\hmQs{\hat{\mQ}^{\ast}}
\newcommand\wB{\widetilde{B}}
\newcommand\mM{\mathcal{M}}

\newcommand\mD{\mathcal{D}}
\newcommand\hmD{\hat{\mathcal{D}}}
\newcommand\wmD{\widetilde{\mathcal{D}}}
\newcommand\mE{\mathcal{E}}
\newcommand\omE{\overline{\mE}}
\newcommand\mF{\mathcal{F}}
\newcommand\hmF{\hat{\mF}}
\newcommand\wmF{\widetilde{\mF}}
\newcommand\mK{\mathcal{K}}
\newcommand\hmK{\hat{\mK}}

\newcommand\mI{\mathcal{I}}
\newcommand\hmL{\hat{\mL}}
\newcommand{\wmL}{\widetilde{\mL}}
\newcommand\mLd{\mathcal{L}^{\dagger}}
\newcommand\hmLd{\hat{\mL}^{\dagger}}
\newcommand\wmLd{\widetilde{\mL}^{\dagger}}
\newcommand{\IM}{\mM^{-1}}

\newcommand{\hf}{\hat{f}}
\newcommand{\tf}{f}
\newcommand{\wf}{\widetilde{f}}

\newcommand{\osigma}{\overline{\sigma}}
\newcommand{\oq}{\overline{q}}

\newcommand\pd[2]{\dfrac{\partial {#1}}{\partial {#2}}}

\newcommand{\UP[1]}{^{(#1)}}
\newcommand{\UB[1]}{^{[#1]}}
\newcommand{\sUP[1]}{^{\ast(#1)}}
\newcommand{\sUB[1]}{^{\ast[#1]}}

\graphicspath{{images/}}

\newtheorem{remark}{Remark}

\title{A Framework of Model Reduction with Arbitrary Orders of Accuracy for the Boltzmann Equation}

\author{
Zhenning Cai\thanks{Department of Mathematics, National University of Singapore, Singapore 119076 email: {\tt
      matcz@nus.edu.sg}.},~~
Ruo Li\thanks{CAPT, LMAM \& School of Mathematical Sciences, Peking University, Beijing, China, 100871; Chongqing Research Institute of Big Data, Peking University, Chongqing, China, 401121, email: {\tt
      rli@math.pku.edu.cn}.},~~
      Yixiao Lu\thanks{HEDPS, CAPT \& School of Mathematical
    Sciences, Peking University, Beijing, China, 100871, email: {\tt
      luyixiao@pku.edu.cn}.},~~
      Yanli Wang\thanks{Beijing Computational
    Science Research Center, email: {\tt ylwang@csrc.ac.cn}.}}

\begin{document}

\maketitle

\begin{abstract}
   This paper presents a general framework for constructing reduced models that approximate the Boltzmann equation with arbitrary orders of accuracy in terms of the Knudsen number $\mathit{Kn}$, applicable to general collision models in rarefied gas dynamics. The framework is based on an orthogonal decomposition of the distribution function into components of different orders in $\mathit{Kn}$, from which the reduced models are systematically derived through asymptotic analysis. Compared to the Chapman-Enskog expansion, our approach yields more tractable model structures. Notably, we establish that a reduced model retaining all terms up to $O(\mathit{Kn}^n)$ in the expansion surprisingly yields models with order of accuracy $O(\mathit{Kn}^{n+1})$. Furthermore, when the collision term is linearized, the accuracy improves dramatically to $O(\mathit{Kn}^{2n})$. These results extend to regularized models containing second-order derivatives. As concrete applications, we explicitly derive 13-moment systems of Burnett and super-Burnett orders valid for arbitrary collision models.
   \vspace*{4mm}

\end{abstract}

\section{Introduction}

The kinetic theory plays a crucial role in statistical mechanics and has attracted significant research interest in recent decades. One important application lies in rarefied gas dynamics, which describes the evolution of dilute gases at mesoscopic scales. The Boltzmann equation, a fundamental model in gas kinetic theory, serves as a core mathematical framework. While progress has been made in numerical simulations of the Boltzmann equation through various approaches \cite{goldstein1989, ZhichengHu2019, Pareschi1996}, these methods are often constrained by the significant complexity of the quadratic collision term.
To address this, several simplified collision models have been proposed as approximations of the original collision term, including the BGK model \cite{BGK}, the ES-BGK model \cite{Holway}, and the Shakhov model \cite{Shakhov}.
compared to classical fluid models, numerical simulations remain computationally expensive due to the additional three-dimensional velocity variable.

Researchers have therefore explored reduced models resembling classical hydrodynamic equations to lower the dimensionality, particularly for gases that are not extremely rarefied. Gas rarefaction is typically characterized by the Knudsen number $\Kn$, defined as the ratio of the mean free path to the characteristic length. Classical macroscopic fluid models, such as the Euler and Navier-Stokes equations, accurately capture fluid behavior in the near-continuum regime ($\Kn < 0.01$) but lose accuracy as the gas becomes more rarefied. For transitional rarefied flows and free molecular flows ($\Kn > 1$), particle methods such as the Direct Simulation Monte Carlo (DSMC) method \cite{Bird} provide efficient simulation approaches. For intermediate regimes ($0.01 < \Kn < 1$), extensive research has focused on developing reduced models.

In \cite{Chapman1916, Enskog1921}, a seminal technique known as the Chapman-Enskog expansion was developed to derive reduced models in the form of conservation equations for mass, momentum, and energy. This method uses asymptotic analysis to generate models approximating the Boltzmann equation to any desired order in terms of the Knudsen number $\Kn$. For instance, the Euler and Navier-Stokes equations correspond to the zeroth-order and first-order models, respectively. Unfortunately, higher-order models, such as the second-order Burnett equations and third-order super-Burnett equations \cite{Burnett1936, Shavaliyiev1993, ReineckeKremer1990}, are not only mathematically complex but also ill-posed due to the loss of linear stability \cite{Bobylev1982}. Although several approaches have been proposed to enhance these models \cite{Bobylev2022, JinSlemrod2001, Bobylev2006}, their practical applications remain severely limited.

An alternative approach of model reduction is the moment method, introduced by Grad in his foundational work \cite{Grad}, which employs Hermite expansions to derive reduced systems.
Moment equations generally exhibit simpler forms and maintain linear stability. Several shortcomings of Grad’s original equations, including the loss of hyperbolicity \cite{Muller1993} and failure to predict smooth shock profiles \cite{Grad1952}, have prompted numerous improvements to broaden the method's applicability \cite{framework, Fox2022, Boehmer2020}. Most of these enhancements focus on reconstructing distribution functions based on given moments, without explicitly considering the accuracy order with respect to $\Kn$.

In \cite{Grad1958}, Grad raised a potential approach to combine the Chapman-Enskog expansion with the moment method, yielding a parabolic system with 13 moments where the moment closure is constructed through asymptotic analysis.
Such an idea was merely a side note in \cite{Grad1958} and was not further developed until 45 years later, when Struchtrup and Torrilhon formulated the regularized 13-moment (R13) equations for Maxwell molecules using the same method \cite{Struchtrup2003}.
Subsequent work \cite{Struchtrup} simplified and rederived these equations using an ``order of magnitude'' approach introduced in \cite{Struchtrup2005R13}.
This technique allows for generalizing R13 equations to arbitrary collision models, and formulates the R13 system for the linear hard-sphere model \cite{Struchtrup2013}.

Recently, a novel methodology has emerged for deriving R13 equations \cite{Cai2024linear, LinR13}, which employs basis functions obtained directly from asymptotic analysis rather than the classical Hermite expansion used in Grad's moment methods.
This offers a straightforward strategy to obtain reduced models with any desired order of accuracy in $\Kn$.
However, existing studies have focused exclusively on the linear regime, leaving the application to the full Boltzmann equation unexplored.

This work extends the framework to the nonlinear Boltzmann equation. We construct a sequence of nested finite-dimensional function spaces $\bbV\UP[0] \subset \bbV\UP[1] \subset \cdots$, ensuring that the projection of the distribution function onto $\bbV\UP[k]$ accurately captures all terms up to $O(\Kn^k)$.
The reduced models are then derived following the methodology of Grad's moment methods, which approximates the Boltzmann equation in these spaces by projections.
Our primary contribution lies in rigorously validating these models' orders of accuracy, which amazingly outperform their apparent precision.
As an application, we demonstrate how this framework can be used to derive 13-moment equations with super-Burnett order of accuracy.

The rest of this paper is organized as follows. Section \ref{sec:background} provides essential background information, including the Chapman-Enskog expansion and moment method for the Boltzmann equation. Our main results, including the derivation of reduced models and their orders of accuracy, are presented in Section \ref{sec:result}. Section \ref{sec:application} illustrates the application of our approach to 13-moment models. Theoretical proofs of the results in Section \ref{sec:result} are provided in Section \ref{sec:validation}. This paper ends with some concluding remarks in Section \ref{sec:conclusion}. 
\section{Chapman-Enskog expansion and moment methods for the Boltzmann equation}
\label{sec:background}
This section reviews classical approaches for deriving reduced models in gas kinetic theory and establishes the notation and working hypotheses used throughout this work.

\subsection{Boltzmann equation}
\label{sec:Boltz}
The gas flow is characterized by a distribution function $f(t,\bx,\bxi)$, dependent on time $t\in\bbR^+$, physical space $\bx\in\bbR^3$, and particle velocity $\bxi\in\bbR^3$. The evolution of this distribution function is governed by the Boltzmann equation:
\begin{equation}
    \label{eq:Boltz}
    \pd{f}{t}+\bxi\cdot \nabla_{\bx}f=\frac{1}{\Kn}\mQ[f,f],
\end{equation}
where $\mQ[f,f]$ denotes the bilinear collision operator and $\Kn$ represents the Knudsen number. The macroscopic equilibrium variables, including density $\rho$, velocity $\bu$ and temperature $\theta$, are defined as
\begin{equation}
    \label{eq:def_conserv}
    \rho(t, \bx)=\int_{\bbR^3}f \,\rd \bxi, \quad
    \bu(t, \bx)=\frac{1}{\rho}\int_{\bbR^3} \bxi f\,\rd \bxi, \quad
    \theta(t, \bx)=\frac{1}{3\rho} \int_{\bbR^3} |\bxi-\bu|^2 f\,\rd\bxi.
\end{equation}
These variables uniquely determine the local Maxwellian distribution $\feq$, satisfying $\mQ[\feq, \feq] = 0$, which is given by:
\begin{equation}
    \label{eq:Maxwellian}
    \feq=\rho\mM,\qquad
    \mM=\frac{1}{(2\pi \theta)^{\frac32}}\exp\left(-\frac{|\bxi-\bu|^2}{2\theta}\right).
\end{equation}

A standard technique in the asymptotic analysis of the Boltzmann equation involves decomposing the collision operator into linear and nonlinear components:
\begin{equation}
    \label{eq:decomp_Q}
    \mQ[f,f]=\mL[f]+\mQs[f,f],
\end{equation}
where
\begin{equation}
    \label{eq:def_L_Qs}
    \begin{split}
        &\mL[f]=\mQ[\feq, f-\feq]+\mQ[f-\feq, \feq]=\mQ[\feq, f]+\mQ[f, \feq]; \\ &\mQs[f,g]=\frac{1}{2} \left(\mQ[f-\feq,g-\feq] + \mQ[g-\feq,f-\feq] \right).
    \end{split}
\end{equation}
This decomposition satisfies $\mL[\feq] = 0$ and $\mQs[\feq,\cdot] = \mQs[\cdot,\feq] = 0$, and reformulates the Boltzmann equation as:
\begin{equation}
    \label{eq:re_Boltz}
    \pd{f}{t}+\bxi\cdot \nabla_{\bx}f=\frac{1}{\Kn}\mL[f]+\frac{1}{\Kn}\mQs[f,f].
\end{equation}
Notably, for a fixed equilibrium state $\feq$, the operator $\mL[f]$ in \eqref{eq:def_L_Qs} constitutes a linear operator acting on $f$.


To facilitate our analysis, we introduce several hypotheses regarding the operators $\mL$ and $\mQ$.
Specifically, we consider $\mL$ as a linear operator on $L^2(\mM^{-1}\rd\bxi)$, parameterized by $\bu$ and $\theta$, where $L^2(\IM\rd\bxi)$ denotes the $L^2$ space with an inner product weighted by a function $\IM$.
Our assumptions are as follows:
\begin{enumerate}[label=\textbf{(H\arabic*)}]

\item \label{Hypo:kerL} The null space of $\mL$ is five-dimensional, corresponding to the conserved quantities:
\begin{equation}
    \label{eq:Ker_L}
    \operatorname{Ker}\mL=\rspan\{\mM, \bxi \mM, |\bxi|^2\mM\}.
\end{equation}
\item \label{Hypo:sym} $\mL$ is symmetric, meaning $\lg \mL f, g\rg_{\IM}=\lg  f, \mL g\rg_{\IM}$, where
\begin{equation}
    \label{eq:def_inpro}
    \lg F, G\rg_{\omega}=\int_{\bbR^3} F G \omega\,\rd \bxi
\end{equation}
denotes the weighted inner product of $L^2(\omega \,\rd \bxi)$. 

\item \label{Hypo:semi-defi} $\mL$ is negative semi-definite on $L^2(\IM \,\rd\bxi)$:
\begin{equation}
    \label{eq:semi-defi}
    \lg \mL f, f \rg_{\IM} \leqslant 0,  
\end{equation}
where equality holds if and only if $f\in \operatorname{Ker} \mL$. 

\item \label{Hypo:L2_Q} $\mQ$ is a bilinear mapping on $L^2(\IM\,\rd\bxi) \times L^2(\IM\,\rd\bxi)\longrightarrow L^2(\IM\,\rd\bxi)$.

\item \label{Hypo:Fredholm} (Fredholm Alternative) $\mL$ is invertible on $(\operatorname{Ker}\mL)^{\perp}$. Consequently, the generalized inverse of $\mL$ is defined as:
\begin{equation}
    \label{eq:def_ldag}
    \mLd(f)=(\mL|_{(\operatorname{Ker} \mL)^{\perp}})^{-1}\left(\mP_{(\operatorname{Ker}\mL)^{\perp}}(f)\right),
\end{equation}
where $\mP_{\bbA}$ denotes orthogonal projection onto subspace $\bbA$. The symmetry of $\mL$ implies that $\mLd$ is likewise symmetric.

\item \label{Hypo:smooth} The operators $\mL$, $\mLd$, and $\mQ$ preserve smoothness with respect to the parameters $\bu$ and $\theta$. Specifically, for any families of functions $\phi_1^{[\bu,\theta]}$ and $\phi_2^{[\bu,\theta]}$ that depend smoothly on $\bu$ and $\theta$, the quantities $\mL[\phi_1^{[\bu,\theta]}]$, $\mLd[\phi_2^{[\bu,\theta]}]$, and $\mQ[\phi_1^{[\bu,\theta]}, \phi_2^{[\bu,\theta]}]$ remain smooth functions of $\bu$ and $\theta$.
\end{enumerate}
The first three assumptions generally hold for all linearized collision operators.
The remaining assumptions, although not universally valid for all collision kernels, are standard in the asymptotic analysis of the Boltzmann equation. 
We will clarify their applicability in the next subsection on the Chapman-Enskog expansion and refer readers to \cite{Baidoo2024} for further discussions of these assumptions.


In certain applications, computational efficiency is improved by neglecting the quadratic term $\mQs$ in the Boltzmann equation \cite{Cai2018,Yin2025}. This modification leads to asymptotic results that differ from the full quadratic case, which will also 
be explored in this work.

\subsection{Chapman-Enskog expansion and Maxwellian iteration}
\label{sec:CE}
The Chapman-Enskog expansion \cite{Chapman1916, Enskog1921} provides a systematic approach to derive moment equations with arbitrary desired order of accuracy in terms of $\Kn$. This method expands the distribution function $f$ as an asymptotic series: 
\begin{equation}
    \label{eq:CE_expan}
    f=\tf\UB[0]+\tf\UB[1]+\cdots,
\end{equation}
where each term $\tf\UB[k] = \Kn^k g\UB[k]$ has magnitude $O(\Kn^k)$ with $g\UB[k] \sim O(1)$. Through the Chapman-Enskog procedure, each $g\UB[k]$ can be expressed functionally in terms of $\rho$, $\bu$, $\theta$, and their derivatives.

The leading-order analysis reveals $\tf\UB[0] = \feq$.
For $k\geqslant 1$, substituting the expansion \eqref{eq:CE_expan} into the Boltzmann equation \eqref{eq:re_Boltz} and matching coefficients of $\Kn^k$ yields:
\begin{equation}
    \label{eq:match_iter}
    \pd{\tf\UB[k-1]}{t}+\bxi\cdot\nabla_{\bx}\tf\UB[k-1]=\frac{1}{\Kn}\mL[\tf\UB[k]]+\frac{1}{\Kn}\mQ\sUB[k],
\end{equation}
where $\mQ\sUB[k]$ collects all $O(\Kn^k)$ contributions from $\mQs$:
\begin{equation}
    \label{eq:exp_Qsupk}
    \mQ\sUB[k]=\sum_{\substack{r,s\geqslant 1 \\ r+s=k} }\mQ[\tf\UB[r], \tf\UB[s]],
\end{equation}
depending only on $\tf\UB[1],\cdots, \tf\UB[k-1]$ and is well-defined due to Hypothesis \ref{Hypo:L2_Q}.
The solution $\tf\UB[k]$ can be determined using the iterative relation:
\begin{equation}
    \label{eq:Maxwell_iter}
    \tf\UB[k]=\Kn\mLd\left(\pd{\tf\UB[k-1]}{t}+\bxi\cdot\nabla_{\bx}\tf\UB[k-1]-\frac{1}{\Kn}\mQ\sUB[k]\right) \in (\ker \mL)^{\perp}.
\end{equation}

The Chapman-Enskog expansion imposes orthogonality to $\ker \mL$ for all $\tf\UB[k]$ ($k \geq 1$), enabling unique determination via Hypothesis \ref{Hypo:Fredholm}. The smoothness of $\tf\UB[k]$ with respect to $\bu$ and $\theta$ is guaranteed by Hypothesis \ref{Hypo:smooth}, making it applicable in subsequent iterations.

This procedure expresses each $\tf\UB[k]$ as a nonlinear function of $\tf\UB[0]$ and its derivatives. In practice, time derivatives are eliminated using conservation laws at a desired order to avoid mixed derivatives.

To construct a system with order of accuracy $O(\Kn^n)$, we truncate the Chapman-Enskog expansion at the $n$-th term and approximate the distribution function by: 
\begin{equation}
    \label{eq:def_f_n}
    f_n=\tf\UB[0]+\tf\UB[1]+\cdots+\tf\UB[n]  
\end{equation} 
We emphasize that $f_n$ is entirely determined by the equilibrium variables $\rho$, $\bu$, and $\theta$. The governing equations for these quantities are derived by projecting the Boltzmann equation onto $\ker\mL$:
\begin{equation}
    \label{eq:norder_approx}
    \mP_{\ker\mL}\left(\pd{f_n}{t}+\bxi\cdot\nabla_{\bx}f_n\right)=0,
\end{equation}
where the right-hand side vanishes due to the conservative property of the collision term. Furthermore, \eqref{eq:norder_approx} can be simplified to:
\begin{equation}
\label{eq:CE-conservation}
    \pd{\feq}{t}+\mP_{\ker\mL} (\bxi\cdot\nabla_{\bx}f_n)=0,
\end{equation}
since $\mP_{\ker \mL} \partial_t f\UB[k] = 0$ for all $k \geq 1$, as guaranteed by the following proposition:

\begin{proposition}
    \label{prop:Ker}
    If $\phi\in \left(\ker\mL\right)^{\perp}$, then $\pd{\phi}{t}, \pd{\phi}{x_i} \in \left(\ker\mL\right)^{\perp},\, i=1,2,3$. 
\end{proposition}
\begin{proof}
    $\phi\in \left(\ker\mL\right)^{\perp}$ is equivalent to $\int_{\bbR^3} (1,\bxi,|\bxi|^2)^T \phi \,\rd \bxi = 0$. Then
    \begin{equation}
        \int_{\bbR^3} (1,\bxi,|\bxi|^2)^T \frac{\partial \phi}{\partial s} \,\rd \bxi = \frac{\partial}{\partial s} \int_{\bbR^3} (1,\bxi,|\bxi|^2)^T \phi \,\rd \bxi = 0, \qquad s=t,x_i.
    \end{equation}
\end{proof}

Specifically, for $n=0,1,2,3$, this procedure yields the well-known Euler, Navier-Stokes, Burnett, and super-Burnett equations, respectively.

An alternative approach to determine $f_n$ is an iterative method known as Maxwellian iteration, obtained by summing \eqref{eq:Maxwell_iter} for $k = 0,\ldots,n$:
\begin{equation}
    \label{eq:re_Maxwell_iter}
    f_n=\feq+\Kn\mLd\left(\pd{f_{n-1}}{t}+\bxi\cdot\nabla_{\bx}f_{n-1}-\frac{1}{\Kn}\mQs_n\right),
\end{equation}
where the cumulative nonlinear term is given by:
\begin{equation}
    \label{eq:def_mQs_l}
    \mQs_n=\sum_{k=1}^{n-1}\mQ\sUB[k]=\sum_{\substack{r,s\geqslant 1 \\r+s\leqslant n}}\mQ[\tf\UB[r], \tf\UB[s]]=\sum_{\substack{r,s\geqslant 1 \\r+s\leqslant n}}\mQ[f_r-f_{r-1}, f_s-f_{s-1}].
\end{equation}
Beginning with $f_0 = \feq$, this recursive procedure generates successive approximations $f_1, f_2, \cdots$.
Originally introduced by Ikenberry and Truesdell \cite{Ikenberry1956} for Maxwell molecules, this iterative technique was later extended to more general molecular interactions \cite{Reinecke1996}, enabling derivations of higher-order conservation laws beyond the Navier-Stokes order \cite{Reinecke1996}.

\subsection{Orders of magnitude}
In the Chapman-Enskog expansion, the claim $\tf\UB[k] \sim O(\Kn^k)$ holds only formally. The solution of the Boltzmann equation is influenced by the Knudsen number $\Kn$, and under certain conditions (e.g., within boundary layers), the scaling of $\tf\UB[k]$ may deviate from $O(\Kn^k)$ due to the direct impact of $\Kn$ on the solution. Despite this, the formal order of accuracy remains fundamental for the development of kinetic models.

To ensure a rigorous analysis, we establish the following rules for handling orders of magnitude:
\begin{itemize}
\item The leading-order term satisfies $\tf\UB[0] = \feq \sim O(1)$.
\item For any function $g \sim O(\Kn^k)$, it follows that $\Kn^{\ell} g \sim O(\Kn^{k+\ell})$ for any integer $\ell$.
\item For any linear operator $\mathcal{A}$ independent of $\Kn$, we have $\mathcal{A} g \sim O(\Kn^k)$ for any $g \sim O(\Kn^k)$.
\item For any two functions $g \sim O(\Kn^k)$ and $h \sim O(\Kn^{\ell})$, the sum satisfies $g + h \sim O(\Kn^{\min(k,\ell)})$.
\item For any bilinear operator $\mathcal{B}$ independent of $\Kn$, it holds that $\mathcal{B}(g,h) \sim O(\Kn^{k+\ell})$ for any two functions $g \sim O(\Kn^k)$ and $h \sim O(\Kn^{\ell})$.
\end{itemize}
In these statements, the operators $\mathcal{A}$ and $\mathcal{B}$ may include differential or integral operations. This means that the potential impact of large gradients in the solution is not explicitly considered when determining the orders of magnitude with respect to $\Kn$. Under this big-O notation, it is formally consistent to state that $f\UB[k] \sim O(\Kn^k)$ within the Chapman-Enskog expansion. All subsequent analyses will adhere to these scaling conventions.

\subsection{Moment models}
An alternative approach to derive reduced models is the moment method, which traces back to Grad’s pioneering work \cite{Grad}. This method involves taking the weighted inner product of the Boltzmann equation with a set of test functions $\phi$, resulting in the moment equations:
\begin{equation}
    \label{eq:basis_moment}
    \left \lg \pd{f}{t}, \phi\right\rg_{\omega}+\left\lg \bxi\cdot \nabla_{\bx}f, \phi\right\rg_{\omega}=\frac{1}{\Kn}\left\lg \mL[f], \phi\right\rg_{\omega}+\frac{1}{\Kn}\left \lg \mQs[f,f], \phi\right\rg_{\omega}.
\end{equation}
Here, $\omega$ is a weight function, and the product $\phi \omega$ is typically chosen as a polynomial of $\bxi$, ensuring that the conservation laws of mass, momentum, and energy are captured in the resulting system.
A moment system is then derived by selecting a finite basis and applying closure relations.

Recent advances \cite{Cai2024linear, LinR13} demonstrate that in the linear regime, appropriate selection of non-polynomial $\phi$ can yield $L^2$-stable moment equations with high asymptotic accuracy using minimal moments. Notably, the 13-moment system can achieve the same asymptotic accuracy as the super-Burnett model. This is made possible by choosing $\phi$ to align with the collision term, relaxing the traditional requirement that $\phi \omega$ must be a polynomial.

We now extend this framework to the nonlinear Boltzmann equation \eqref{eq:re_Boltz} with expanded moment systems.
As noted in Section \ref{sec:Boltz}, we also consider the case of linearized collision, where the quadratic part $\mQs$ vanishes in \eqref{eq:basis_moment}.

\section{Main results}
\label{sec:result}
In general, the model \eqref{eq:basis_moment} with a finite set of indices $\alpha$ is not inherently closed.
This means the resulting system cannot be directly formulated as a finite set of equations for functions of $\bx$ and $t$,
with their evolution fully described. Grad's closure method addresses this by approximating the distribution function through a finite expansion:
\begin{equation} 
    \label{eq:Grad_ansatz}
    f_{\mathrm{Grad}}(t,\bx,\bxi) = \sum_{\alpha \in \mathscr{I}} f_{\alpha}(t,\bx) p_{\alpha}(t,\bx,\bxi) \mM(t,\bx,\bxi),
\end{equation}
where $p_{\alpha}$ are Hermite polynomials in $\bxi$ (translated and scaled by local velocity and temperature), and $\mathscr{I}$ is an index set whose cardinality matches the number of equations. This ansatz converts \eqref{eq:basis_moment} into equations for the coefficients $f_{\alpha}$.

Compared to the Chapman-Enskog expansion, Grad's approach offers two advantages: (1) direct formulation without asymptotic series, and (2) lower-order derivatives in the resulting equations.
However, for most collision models (excepting Maxwell molecules and certain BGK-type models), the approximation error $f - f_{\mathrm{Grad}}$ remains $O(\Kn)$, regardless of the size of $\mathscr{I}$. 
Consequently, even the Navier-Stokes equations cannot be accurately derived from Grad's moment equations without parameter adjustments.

This section presents our main contributions. Specifically, we introduce a novel choice of basis functions
in the ansatz \eqref{eq:Grad_ansatz} that allows for approximating $f$ with any desired asymptotic order.
This enables the development of models with arbitrarily high order of accuracy while maintaining low-order derivatives.

\subsection{Decomposition of the function space}
\label{sec:order}
We employ a similar ansatz to Grad's method but project onto a carefully constructed function space. Our framework adopts the following mathematical conventions:

\begin{itemize}
\item All blackboard-bold spaces (e.g., $\bbV$) are Hilbert spaces with inner product $\langle \cdot, \cdot \rangle_{\mM^{-1}}$. This choice aligns with Grad's moment method and allows us to exploit the symmetry of the operator $\mL$ (see Hypothesis \ref{Hypo:sym}).
\item Since $\mM^{-1}$ depends on $\bu$ and $\theta$ (see \eqref{eq:Maxwellian}), all spaces inherit these parameters. Set relationships and operations (such as `$\subset$', `$\oplus$', `$\cap$') between spaces hold for all valid $\bu$ and $\theta$.
\item Each basis function $\vp$ represents a parameterized family with respect to $\bu$ and $\theta$. We denote parameter derivatives as $\partial_{u_i} \vp$ and $\partial_{\theta} \vp$.
\item Our analysis assumes that $f$ is the solution of the Boltzmann equation, with 
$\mM$ representing the corresponding local Maxwellian. For a given function $g(t,\bx,\bv)$, we say $g \in \mathbb{V}$ if $g(t,\bx,\cdot)$ is a member of $\bbV$ for all $(t,\bx)$, where $\bbV$ is parameterized by $\bu(t,\bx)$ and $\theta(t,\bx)$ from $\mM(t,\bx,\cdot)$.
\end{itemize}

We now construct a hierarchy of function spaces $\bbV^{(k)}$ designed to capture all $O(\Kn^k)$ terms in the distribution function:

\begin{definition}
Let
\begin{align}
\bbW\UP[0]&=\{\mM\},\quad \bbV\UP[0]=\operatorname{Ker}\mL=\operatorname{span}\{\mM, \bxi\mM, |\bxi|^2\mM \}, \label{eq:def_V0}\\
    \begin{split}
    \label{eq:def_F0}
    \bbF\UP[0]&=\Big\{\mM,\, \partial_{u_i}\mM,\, \partial_{\theta}\mM,\, \xi_j\mM,\, 
        \xi_j\partial_{u_i}\mM,\, \xi_j\partial_{\theta}\mM \Big| i,j=1,2,3\Big\} \\
        &=\Big\{\mM,\, \xi_i\mM,\, \xi_i\xi_j\,\mM,\, \xi_j|\bxi|^2\mM\Big| i,j=1,2,3 \Big\}.
    \end{split}
\end{align}
For $k \geqslant 1$, the function sets $\bbW\UP[k]$, $\bbF\UP[k]$, and $\bbV\UP[k]$ are defined recursively:
\begin{equation}
    \label{eq:def_VF}
    \begin{split}
    \bbW\UP[k] & =\mLd\left(\bbF\UP[k-1]\right)\cup \{\mM\}, \quad
    \bbV\UP[k] = \rspan(\bbW\UP[k])\oplus \bbV\UP[0], \\
    \bbF\UP[k] &= \bigg\{ \vp\UP[k],\, \partial_{u_i} \vp\UP[k],\, \partial_{\theta} \vp\UP[k],\,\xi_j\vp\UP[k],\, \xi_j \partial_{u_i} \vp\UP[k],\, \xi_j \partial_{\theta} \vp\UP[k],\, \mQ^*[\vp\UP[r], \vp\UP[s]] \,\Big| \\
    &\qquad\qquad\qquad \vp\UP[l] \in \bbW\UP[l], \, r+s \leqslant k+1, \, i,j=1,2,3 \bigg\}.
\end{split}
\end{equation}
\end{definition}

In this construction, $\bbV\UP[k]$ is a linear function space, while $\bbW\UP[k]$ and $\bbF\UP[k]$ are finite sets of functions (see Theorem \ref{thm:order} below).  
Besides, the hierarchy is nested: 
$$\bbW\UP[k]\subset \bbW\UP[l],\,\bbV\UP[k]\subset \bbV\UP[l],\,\bbF\UP[k]\subset \bbF\UP[l], \quad k\leqslant l.$$



The following theorem establishes the crucial relationship between the function spaces $\bbV\UP[k]$ and the asymptotic behavior of the distribution function:
\begin{theorem}
\label{thm:order}
For any $k\geqslant 0$, it holds:
\begin{enumerate}
    \item \label{state:finite} Both $\bbW\UP[k]$ and $\bbF\UP[k]$ have finite elements, and thus $\bbV\UP[k]$ has finite dimensions.
    \item For any $\psi\UP[k+1]$ satisfying $\psi\UP[k+1]\perp \bbV\UP[k]$ under the inner product $\lg\cdot, \cdot\rg$, it holds $\lg \psi\UP[k+1], f\rg \sim O(\Kn^{k+1})$. \label{state:order}
\end{enumerate}
    
\end{theorem}

The detailed proof of this theorem is provided in Section \ref{sec:proof_thm_order}. It demonstrates that the projection $\mP_{\bbV\UP[k]}{f}$ retains all $O(\Kn^k)$ information of the exact solution $f$. This key insight motivates our model reduction strategy: by projecting onto $\bbV\UP[k]$, we obtain a finite-dimensional system while preserving the desired asymptotic accuracy. The remainder of this section develops this approach and rigorously analyzes their orders of accuracy in terms of $\Kn$.

\begin{remark}
The function space construction directly corresponds to the Maxwellian iteration process. From \eqref{eq:re_Maxwell_iter}, we have:
\begin{displaymath}
\pd{f_{n-1}}{t}+\bxi\cdot\nabla_{\bx}f_{n-1}-\frac{1}{\Kn}\mQs_n \in \rspan(\bbF\UP[n-1]), 
\qquad f_n \in \rspan(\bbW\UP[n]).
\end{displaymath}
While $\bbV\UP[n]$ is designed such that $f_n \in \bbV\UP[n]$, it is intentionally defined to be slightly larger than $\rspan(\bbW\UP[n])$. This can be observed from \eqref{eq:def_V0}, where $\bbV\UP[0]$ has four more basis functions compared to $\rspan(\bbW\UP[0])$.
These additional dimensions allow $\bbV\UP[k]$ to serve as a basis for reduced models, which naturally accommodates conservation laws for mass, momentum, and energy. Further details will be presented in the following subsection.
\end{remark}

\subsection{Derivation of reduced models}
\label{sec:deri_moment}
For notational simplicity, we write $\langle \cdot, \cdot \rangle_{\IM}$ as $\langle \cdot, \cdot \rangle$ in the following.  
Theorem \ref{thm:order} guarantees that the projection $f^{(k)} := \mP_{(\bbV\UP[k-1])^{\perp} \cap \bbV\UP[k]} f$ satisfies $f^{(k)} \sim O(\Kn^k)$. This motivates the decomposition:
\begin{equation}
    \label{eq:f2fk}
    f=\sum_{k=0}^{\infty}f\UP[k].
\end{equation}
Let $\{\vp_{\alpha}\UP[k]\}$ denote an orthogonal basis for $(\bbV\UP[k-1])^{\perp} \cap \bbV\UP[k]$ (with $\bbV\UP[-1] := \{0\}$). Then, $f\UP[k]$ can be further expanded as
\begin{equation}
    \label{eq:exp_further_fk}
    f\UP[k]=\sum_{\alpha=1}^{N_k} f_{\alpha}\UP[k] \vp_{\alpha}\UP[k], \qquad
    N_k = \dim\left((\bbV\UP[k-1])^{\perp} \cap \bbV\UP[k]\right),
\end{equation}
where $f_{\alpha}\UP[k] = C_{\alpha}\UP[k] \langle f, \vp_{\alpha}\UP[k] \rangle$ with $C_{\alpha}\UP[k] = 1/\langle \vp_{\alpha}\UP[k], \vp_{\alpha}\UP[k] \rangle$.

Notably, the zeroth-order term $f\UP[0]$ coincides with the local Maxwellian $\feq$ or $f\UB[0]$ from the Chapman-Enskog expansion. When $k > 0$, while both $f\UP[k]$ and $\tf\UB[k]$ are $O(\Kn^k)$, they differ by higher-order terms due to their distinct constructions.

The derivation begins by projecting the Boltzmann equation onto the equilibrium subspace $\bbV\UP[0]$ (defined in \eqref{eq:def_V0}) to obtain the fundamental conservation laws, which is independent of the specific form of the collision operator. We select the following orthogonal basis for this five-dimensional space:
\begin{equation} \label{eq:V0basis}
\varphi_1\UP[0] = \mM, \,\,
\varphi_{i+1}\UP[0] = (\xi_i - u_i) \mM, \,\, i=1,2,3; \,\,
\varphi_5\UP[0] = \frac{|\bxi-\bu|^2-3\theta}{2} \mM.
\end{equation}
Projecting the Boltzmann equation onto these basis functions leads to the conservation laws:
\begin{equation}
    \label{eq:conserv_law}
    \begin{split}
        &\pd{\rho}{t}+\bu\cdot\nabla_{\bx}\rho+\rho\nabla_{\bx}\cdot{\bu}=0, \\
        &\rho\pd{u_i}{t}+\rho\bu\cdot\nabla_{\bx}u_i+\theta\pd{\rho}{x_i}+\rho\pd{\theta}{x_i}+ 
        \sum_{l=1}^{\infty}\left\lg\bxi\cdot\nabla_{\bx}f\UP[l], (\xi_i-u_i)\mM\right\rg=0, \\ 
        &\frac32\rho\pd{\theta}{t}+\frac32\rho\bu\cdot\nabla_{\bx}{\theta}+\rho\theta\nabla_{\bx}\cdot\bu+ 
        \sum_{l=1}^{\infty}\left\lg\bxi\cdot\nabla_{\bx}f\UP[l], \left(\frac{|\bxi-\bu|^2}{2}-\frac{3\theta}2\right)\mM\right\rg=0.
    \end{split}
\end{equation}
These equations are derived using the expansion \eqref{eq:f2fk} with $f\UP[0] = \rho \mM$, and leveraging the fact that $\mL[f]$, $\mQs[f,f]$, $\partial_t f\UP[k]$ and $\partial_{x_i}f\UP[k]$ are all orthogonal to $\bbV\UP[0]$ when $k \geqslant 1$ (see Proposition \ref{prop:Ker}). 
The term $f\UP[k]$ can be expressed via the coefficients $f_{\alpha}\UP[k]$ using \eqref{eq:exp_further_fk}.


To derive the evolution equations for the general coefficients $f_{\alpha}\UP[k]$, we substitute the test function $\phi = \vp_{\alpha}\UP[k]$ into the moment equation \eqref{eq:basis_moment} for $k\geqslant 1$ and $1\leqslant\alpha\leqslant N_k$. Using the expansions \eqref{eq:f2fk} and \eqref{eq:exp_further_fk}, we obtain the coupled system:
\begin{equation}
    \label{eq:simp_moment}
    \pd{f_{\alpha}\UP[k]}{t}+\sum_{l=0}^{+\infty} \sum_{\beta=1}^{N_l} \left( \bA_{\alpha,\beta}\UP[kl]\cdot\nabla_{\bx}f_{\beta}\UP[l]+B_{\alpha,\beta}\UP[kl]f_{\beta}\UP[l] \right)=\frac{1}{\Kn}\sum_{l=1}^{+\infty} \sum_{\beta=1}^{N_l} L_{\alpha,\beta}\UP[kl]f_{\beta}\UP[l]+\frac{1}{\Kn} \mQ_{\alpha}\sUP[k],
\end{equation}
where 
\begin{equation}
    \label{eq:Q_alpha}
    \mQ_{\alpha}\sUP[k]=\sum_{l,m=1}^{+\infty} \sum_{\beta=1}^{N_l} \sum_{\gamma=1}^{N_m} Q_{\alpha,\beta,\gamma}\UP[klm] f_{\beta}\UP[l] f_{\gamma}\UP[m].
\end{equation}
and
\begin{equation}
    \label{eq:notation_ABL}
    \begin{gathered}
        \bA_{\alpha,\beta}\UP[kl]=C_{\alpha}\UP[k] \left\lg\vp_{\alpha}\UP[k] , \bxi\vp_{\beta}\UP[l]\right\rg, \qquad
        B_{\alpha,\beta}\UP[kl]=C_{\alpha}\UP[k] \left\lg\vp_{\alpha}\UP[k], \pd{\vp_{\beta}\UP[l]}{t}+\bxi\cdot\nabla_{\bx}\vp_{\beta}\UP[l]\right\rg, \\
        L_{\alpha,\beta}\UP[kl]=C_{\alpha}\UP[k]\left\lg\vp_{\alpha}\UP[k], \mL[\vp_{\beta}\UP[l]]\right\rg, \qquad
        Q_{\alpha,\beta,\gamma}\UP[klm] = C_{\alpha}\UP[k]\left\lg\vp_{\alpha}\UP[k], \mQ[\vp_{\beta}\UP[l],\vp_{\gamma}\UP[m]]\right\rg.
    \end{gathered}
\end{equation}
The derivation incorporates the fact $\mL(f\UP[0]) = 0$ and 
\begin{equation}
    \label{eq:mQs_sum}
    \mQs[f,f]=\mQ[f-f\UP[0], f-f\UP[0]]=\sum_{l,m=1}^{\infty}\mQ[f\UP[l],f\UP[m]],
\end{equation}
so the indices $l$ and $m$ in the summations of $\mL$ and $\mQ$ begin with $1$ rather than $0$. 

It is worth noting that setting $k = 0$ in \eqref{eq:simp_moment} recovers the conservation laws in \eqref{eq:conserv_law}, but these equations have a distinctive nature. Specifically, in this case, $f_2\UP[0] = f_3\UP[0] = f_4\UP[0] = f_5\UP[0] \equiv 0$ due to the choice of basis functions in \eqref{eq:V0basis}. Consequently, when $k = 0$, the equations \eqref{eq:simp_moment} for $\alpha = 2,3,4,5$ do not govern the dynamics of $f_{\alpha}\UP[k]$. Instead, they dictate the evolution of the velocity $\bu$ and the temperature $\theta$, whose time derivatives appear in the coefficients $B_{\alpha,\beta}\UP[kl]$.

Considering all $k \in \mathbb{N}$ and $\alpha = 1,\cdots,N_k$, the equations in \eqref{eq:simp_moment} constitute an infinite system. Our reduced models are derived by selectively removing certain terms from \eqref{eq:conserv_law} and \eqref{eq:simp_moment}. In this study, we will explore two types of reduced models.

\subsubsection{Hyperbolic reduced models}
The first type of reduced model can be viewed as a generalization of the Euler equations and Grad's moment equations, characterized by a first-order quasi-linear system.
Given a positive integer $n$, this system is obtained by retaining only the equations derived from the inner product with $\varphi_{\alpha}\UP[k]$ for $k \leqslant n$, along with the following closure conditions:
\begin{enumerate}
\item In \eqref{eq:simp_moment}, set $f_{\alpha}\UP[k] = 0$ for all $k > n$.
\item In \eqref{eq:Q_alpha}, discard all terms with $l + m > n + 1$.
\end{enumerate}
This leads to the following reduced model formulation:
\begin{equation} \label{eq:closure}
  \begin{split}
    &\pd{f_{\alpha}\UP[k]}{t}+\sum_{l=0}^n \sum_{\beta=1}^{N_l} \left( \bA_{\alpha,\beta}\UP[kl]\cdot\nabla_{\bx}f_{\beta}\UP[l]+B_{\alpha,\beta}\UP[kl]f_{\beta}\UP[l] \right) \\
    &\quad =\frac{1}{\Kn}\sum_{l=1}^n \sum_{\beta=1}^{N_l} L_{\alpha,\beta}\UP[kl]f_{\beta}\UP[l]+\frac{1}{\Kn} \hmQ_{\alpha}\sUP[k], \qquad 
    0\leqslant k\leqslant n, \quad 1\leqslant \alpha\leqslant N_k,
  \end{split}
\end{equation}
where
\begin{equation}
    \label{eq:def_hmQ}
    \hmQ_{\alpha}\sUP[k]=\left\lg \vp_{\alpha}\UP[k], \sum_{l=1}^{n} \sum_{m=1}^{n+1-l}\mQ[f\UP[l],f\UP[m]]\right\rg=\sum_{l=1}^{n-1} \sum_{m=1}^{n+1-l} \sum_{\beta=1}^{N_l} \sum_{\gamma=1}^{N_m} Q_{\alpha,\beta,\gamma}\UP[klm] f_{\beta}\UP[l] f_{\gamma}\UP[m].
\end{equation}
In the case of the linearized collision operator, where $\mQs$ is absent, the closure retains the form of \eqref{eq:closure}, but all coefficients $Q_{\alpha,\beta,\gamma}\UP[klm]$ are set to zero.

The time derivatives appearing in $B_{\alpha,\beta}\UP[kl]$ require careful handling through the conservation laws. We define the spatial flux operators:
\begin{equation}
\label{eq:def_mE}
\begin{split}
&\mathcal{E}_{u_i}\UP[0] = -\bu \cdot \nabla_{\bx} u_i - \frac{\theta}{\rho} \pd{\rho}{x_i} - \pd{\theta}{x_i}, \qquad \mathcal{E}_{\theta}\UP[0] = -\rho \bu \cdot \nabla_{\bx} \theta - \frac{2}{3} \theta \nabla \cdot \bu, \\
&\mathcal{E}_{u_i}\UP[j] = -\frac{1}{\rho}\left\lg \bxi \cdot \nabla_{\bx} f\UP[j], (\xi_i - u_i) \mM \right\rg, \\
&\mathcal{E}_{\theta}\UP[j] = -\frac{2}{3\rho}\sum_{\alpha=1}^{N_j} \left\lg \bxi \cdot \nabla_{\bx} f\UP[j], \left( \frac{|\bxi-\bu|^2}{3} - \theta \right) \mM \right\rg, \qquad j = 1,2,\ldots,
\end{split}
\end{equation}
where $f\UP[j]$ can be further expanded using \eqref{eq:exp_further_fk}. The time derivatives can thus be expressed as:
\begin{equation}
    \label{eq:re_conserv}
    \pd{u_i}{t}=\sum_{j=0}^n \mE_{u_i}\UP[j],\qquad \pd{\theta}{t}=\sum_{j=0}^n \mE_{\theta}\UP[j].
\end{equation}
Consequently, the time derivatives in $B_{\alpha,\beta}\UP[kl]$ (see \eqref{eq:notation_ABL}) can be computed using
\begin{equation} \label{eq:time_derivative}
\pd{\vp_{\beta}\UP[l]}{t}=\sum_{i=1}^3\pd{\vp_{\beta}\UP[l]}{u_i}\pd{u_i}{t}+\pd{\vp_{\beta}\UP[l]}{\theta}\pd{\theta}{t}, 
\end{equation}
leading to the form
\begin{equation} \label{eq:B}
\begin{split}
B_{\alpha,\beta}\UP[kl] &= C_{\alpha}\UP[k] \left[ \sum_{i=1}^3 \left\lg \vp_{\alpha}\UP[k], \pd{\vp_{\beta}\UP[l]}{u_i} \right\rg \pd{u_i}{t} + \left\lg \vp_{\alpha}\UP[k], \pd{\vp_{\beta}\UP[l]}{\theta} \right\rg \pd{\theta}{t} + \left\lg \vp_{\alpha}\UP[k], \bxi \cdot \nabla_{\bx} \vp_{\beta}\UP[l]\right\rg \right] \\
&= C_{\alpha}\UP[k] \left[ \sum_{j=0}^n \left( \sum_{i=1}^3 \left\lg \vp_{\alpha}\UP[k], \pd{\vp_{\beta}\UP[l]}{u_i} \right\rg \mathcal{E}_{u_i}\UP[j] + \left\lg \vp_{\alpha}\UP[k], \pd{\vp_{\beta}\UP[l]}{\theta} \right\rg \mathcal{E}_{\theta}\UP[j] \right) + \left\lg \vp_{\alpha}\UP[k], \bxi \cdot \nabla_{\bx} \vp_{\beta}\UP[l]\right\rg \right],
\end{split}
\end{equation}
where the spatial derivative $\nabla_{\bx} \vp_{\beta}\UP[l]$ can be simplified as
$$
\nabla_{\bx} \vp_{\beta}\UP[l] = \sum_{i=1}^3\pd{\vp_{\beta}\UP[l]}{u_i} \nabla_{\bx} u_i+\pd{\vp_{\beta}\UP[l]}{\theta}\nabla_{\bx} \theta.
$$
For $k = 1, \cdots, n$, plugging \eqref{eq:B} into \eqref{eq:closure} yields the evolution equations for $f_{\alpha}\UP[k]$, which involve only first-order spatial derivatives.

These resulting equations exhibit a hyperbolic form, but their hyperbolicity is not universally guaranteed. For Maxwell molecules, the system reduces to Grad's moment equations, which are only hyperbolic near equilibrium \cite{Cai2013}. This limitation motivates the introduction of second-order derivatives into our models as regularizations.

\subsubsection{Regularized reduced models}
\label{sec:regularized}
The second type of models generalize the Navier-Stokes-Fourier equations and regularized moment equations, which have second-order derivatives in some of the equations. Such models are more likely to produce smooth solutions, thereby avoiding non-physical subshocks in the computation of shock structures \cite{Torrilhon2004}. For these models, the equations for $f_{\alpha}\UP[k]$ with $k = 0,1,\cdots,n-1$ are given by
\begin{equation} \label{eq:closure2}
  \begin{split}
    & \pd{f_{\alpha}\UP[k]}{t}+\sum_{l=0}^n \sum_{\beta=1}^{N_l} \bA_{\alpha,\beta}\UP[kl]\cdot\nabla_{\bx}f_{\beta}\UP[l]+\sum_{l=0}^{n} \sum_{\beta=1}^{N_l} \wB_{\alpha,\beta}\UP[kl]f_{\beta}\UP[l] \\
    & \quad =\frac{1}{\Kn}\sum_{l=1}^n \sum_{\beta=1}^{N_l} L_{\alpha,\beta}\UP[kl]f_{\beta}\UP[l]+\frac{1}{\Kn} \widetilde{\mQ}_{\alpha}\sUP[k], \quad 0\leqslant k\leqslant n-1, \quad 1\leqslant \alpha\leqslant N_k,
  \end{split}
\end{equation}
where 
\begin{equation}
    \label{eq:def_wB}
    \begin{aligned}
    \wB_{\alpha,\beta}\UP[kl] &= C_{\alpha}\UP[k] \Bigg[ \sum_{j=0}^{n-1} \left( \sum_{i=1}^3 \left\lg \vp_{\alpha}\UP[k], \pd{\vp_{\beta}\UP[l]}{u_i} \right\rg \mathcal{E}_{u_i}\UP[j] + \left\lg \vp_{\alpha}\UP[k], \pd{\vp_{\beta}\UP[l]}{\theta} \right\rg \mathcal{E}_{\theta}\UP[j] \right) \\
    \qquad\qquad\qquad&+ \left\lg \vp_{\alpha}\UP[k], \bxi \cdot \nabla_{\bx} \vp_{\beta}\UP[l]\right\rg \Bigg], \quad k = n \text{\quad or\quad} k > 0, l = n; \\
    \wB_{\alpha,\beta}\UP[kl] &= B_{\alpha,\beta}\UP[kl], \quad \text{otherwise,} 
    \end{aligned}
\end{equation}
and
\begin{equation}
    \label{eq:def_wmQ}
    \begin{aligned}
    \widetilde{\mQ}_{\alpha}\sUP[0] &= 0, \qquad \widetilde{\mQ}_{\alpha}\sUP[1]=
    \begin{cases}
    \hmQ_{\alpha}\sUP[1], & n \geqslant 2, \\
    0, & n = 1;
    \end{cases}
    \\
    \widetilde{\mQ}_{\alpha}\sUP[k]&=\sum_{l=1}^{n-1} \sum_{m=1}^{n-l} \sum_{\beta=1}^{N_l} \sum_{\gamma=1}^{N_m} Q_{\alpha,\beta,\gamma}\UP[klm] f_{\beta}\UP[l] f_{\gamma}\UP[m], \qquad 2 \leqslant k \leqslant n.
    \end{aligned}
\end{equation}
To close the system, an expression for $f_{\beta}\UP[n]$ is required, which is determined by
\begin{equation} \label{eq:fn}
    \begin{aligned}
    \sum_{\beta=1}^{N_n} L_{\alpha,\beta}\UP[nn] f_{\beta}\UP[n] = \sum_{l=0}^{n-1} \sum_{\beta=1}^{N_l} \Kn \left( \bA_{\alpha,\beta}\UP[nl]\cdot\nabla_{\bx}f_{\beta}\UP[l]+\widetilde{B}_{\alpha,\beta}\UP[nl]f_{\beta}\UP[l] - \frac{1}{\Kn} L_{\alpha,\beta}\UP[nl] f_{\beta}\UP[l] \right) - \widetilde{\mQ}_{\alpha}\sUP[n], \\
    \alpha = 1,\cdots,N_n.
    \end{aligned}
\end{equation}
This equation represents a linear system that must be solved to express $f_{\beta}\UP[n]$ in terms of $f_{\beta}\UP[k]$ for $k = 0,1,\cdots,n-1$. Notably, $f_{\beta}\UP[n]$ involves at most first-order spatial derivatives. 
Consequently, when these expressions are substituted back into \eqref{eq:closure2}, only second-order derivatives are introduced through terms involving $\nabla_{\bx} f_{\beta}\UP[n]$, $\mathcal{E}_{u_i}\UP[n]$, and products like $\widetilde{B}_{\alpha,\beta}\UP[kn]f_{\beta}\UP[n]$. 

Comparing the hyperbolic model \eqref{eq:closure} with the regularized model defined by \eqref{eq:closure2} and \eqref{eq:fn}, it is evident that the latter omits several terms present in the former. Most significantly, the time derivative of $f_{\alpha}\UP[n]$ is eliminated in \eqref{eq:fn}, fundamentally altering the nature of the equations. These models include the regularized 13-moment equations for Maxwell molecules \cite{Struchtrup2003}, which is why we refer to them as \emph{regularized reduced models}.

\subsection{Asymptotic accuracy of the reduced models}
\label{sec:thm}
The Maxwellian iteration technique from Section \ref{sec:CE} can also be applied to the reduced models \eqref{eq:closure}, yielding conservation laws analogous to \eqref{eq:CE-conservation}.
However, due to the truncation of the series by setting $f_{\alpha}\UP[k] = 0$ for all $k > n$, the resulting conservation laws may differ slightly from those directly derived from the Boltzmann equation.

Following \cite[Chapter 8]{Struchtrup}, we define the order of accuracy for a reduced model by examining the discrepancy as follows:
\begin{definition}[Order of accuracy]
\label{def:OOA}
A reduced model of the Boltzmann equation is said to have an order of accuracy $O(\Kn^n)$ if its Maxwellian iteration produces conservation laws matching the exact Boltzmann results through $O(\Kn^n)$ terms.
\end{definition}

For instance, the Euler equations, corresponding to the case $n = 0$ in \eqref{eq:closure}, have the zeroth-order accuracy.
In the case of Maxwell molecules, setting $n = 1$ produces Grad's 13-moment equations. It has been shown in \cite{Struchtrup2004} that the exact Burnett equations can be derived from these 13-moment equations through asymptotic analysis, indicating that Grad's 13-moment equations for Maxwell molecules possess an order of accuracy $O(\Kn^2)$.

At first glance, the systems \eqref{eq:closure} and \eqref{eq:closure2}\eqref{eq:fn} appear to have order of accuracy $O(\Kn^n)$, since the conservation laws in \eqref{eq:conserv_law} retain all terms up to the $n$-th order.
Surprisingly, these systems actually achieve the order of accuracy $O(\Kn^{n+1})$, as stated in the following theorem:
\begin{theorem}
    \label{thm:full}
    When $n\geqslant 2$, both the hyperbolic reduced model \eqref{eq:closure} and the regularized reduced model  \eqref{eq:closure2}\eqref{eq:fn} for the Boltzmann equation have an order of accuracy $O(\Kn^{n+1})$. 
\end{theorem}

Furthermore, a higher order of accuracy can be achieved when the collision operator is linearized:
\begin{theorem}
    \label{thm:linear}
    When $\mQ^*[\cdot,\cdot] = 0$, the model \eqref{eq:closure} has order of accuracy $O(\Kn^{2n})$, and the model \eqref{eq:closure2}\eqref{eq:fn} has order of accuracy $O(\Kn^{2n-1})$.
\end{theorem}
The proofs of Theorems \ref{thm:full} and \ref{thm:linear} will be presented in Sections \ref{sec:lin_order} and \ref{sec:full_order}.

We now highlight two notable special cases of the above theorems:
\begin{itemize}
    \item When $n = 1$, the hyperbolic models for both linearized and quadratic collision operators exhibit the same order of accuracy, $O(\Kn^2)$, commonly known as the ``Burnett order" in the context of Chapman-Enskog theory.
    \item When $n = 2$, the regularized models for both types of collision operators achieve an order of accuracy of $O(\Kn^3)$, referred to as the ``super-Burnett order".
\end{itemize}
In both cases, the systems essentially contain 13 variables, a configuration that has been widely studied in the literature. We will provide a more detailed discussion of these 13-moment models in the next section.
It is also worth noting that for larger values of $n$, the orders of accuracy for the linear and nonlinear models will begin to diverge.

\section{Applications to the 13-moment model}
\label{sec:application}

The original 13-moment model proposed by Grad \cite{Grad} consisted of five equilibrium variables plus their fluxes, including the stress tensor $\sigma_{ij}$ and heat flux $q_i$.
Subsequent works have explored different choices for these 13 variables while always retaining the five conservative moments \cite{Cai2014, Struchtrup2022, LinR13}. 

In Grad's approach, the distribution function is approximated as:
\begin{displaymath}
f \approx \rho \mathcal{M} + \sum_{i,j=1}^3 \frac{\sigma_{ij}}{2\theta^2} \varphi_{ij}\UP[1] + \sum_{i=1}^3 \frac{2}{5} \frac{q_i}{\theta^2} \varphi_i\UP[1],
\end{displaymath}
where the basis functions are defined by
\begin{displaymath}
\varphi_{ij}\UP[1] = \left[ (\xi_i - u_i) (\xi_j - u_j) - \frac{1}{3} \delta_{ij} |\bxi-\bu|^2 \right] \mM, \qquad
\varphi_i\UP[1] = (\xi_i - u_i) \left[ \frac{|\bxi-\bu|^2}{2\theta} - \frac{5}{2} \right] \mM,
\end{displaymath}
with $\sigma_{ij} = \lg f, \varphi_{ij}\UP[1] \rg$, $q_i = \lg f, \varphi_i\UP[1] \rg$.
The GENERIC-13 model \cite{Struchtrup2022} generalizes this approach by introducing a temperature tensor $\Theta_{ij}$, leading to a modified definition of the moments. 

Our work follows \cite{LinR13}, where the distribution function $f$ is approximated through its projection onto the space $\bbV\UP[1]$.  
To construct the moment system using this method, we first establish a set of basis functions for $\bbV\UP[1]$ and confirm its 13-dimensional structure. By \eqref{eq:def_F0} and \eqref{eq:def_VF}, direct computation yields
\begin{align}
\bbW\UP[1] &= \mLd(\bbF\UP[0]) \cup \{\mM\} = \{\mM\} \cup \left\{\mLd(\xi_i\xi_j\mM),\,\mLd(\xi_i|\bxi|^2\mM)\,\Big|\, i,j=1,2,3\right\}, \\
\label{eq:exp_V1}
\bbV\UP[1] &= \rspan(\bbW\UP[1]) \oplus \bbV\UP[0] = \rspan\left\{\mM,\, \xi_i\mM,\, |\bxi|^2\mM,\, \phi_{ij}\UP[1],\, \phi_{i}\UP[1]\,\Big|\,i,j=1,2,3\right\},
\end{align}
where
\begin{equation} \label{eq:V1basis}
    \phi_{ij}\UP[1]=C_2 \mLd\varphi_{ij}\UP[1], \quad
    \phi_{i}\UP[1]=C_1 \mLd \varphi_i\UP[1], \qquad i,j=1,2,3.
\end{equation}
The normalization constants $C_2$ and $C_1$ are chosen such that
\begin{displaymath}
\langle \phi_{ij}\UP[1], \phi_{ij}\UP[1] \rangle_{\IM} = \langle \varphi_{ij}\UP[1], \varphi_{ij}\UP[1] \rangle_{\IM}, \quad
\langle \phi_i\UP[1], \phi_i\UP[1] \rangle_{\IM} = \langle \varphi_i\UP[1], \varphi_i\UP[1] \rangle_{\IM}, \quad i,j=1,2,3.
\end{displaymath}
The traceless symmetry $\phi_{ij}\UP[1] = \phi_{ji}\UP[1]$ and $\sum_{i=1}^3 \phi_{ii}\UP[1] = 0$ ensures $\dim \bbV\UP[1] = 13$.

Thus, the variables in the reduced model are $\rho, u_i, \theta, \osigma_{ij}, \oq_i$, where
\begin{equation}
    \label{eq:def_osigma_oq}
    \osigma_{ij}=\lg f,\, \phi_{ij}\UP[1]\rg,\qquad \oq_i=\lg f,\, \phi_{i}\UP[1]\rg, \qquad i,j = 1,2,3,
\end{equation}
satisfying $\osigma_{ij} = \osigma_{ji}$ and $\osigma_{11} + \osigma_{22} + \osigma_{33} = 0$.
By Theorem \ref{thm:order}, it holds $\rho, \bu, \theta\sim O(1)$ and $\osigma_{ij}, \oq_i\sim O(\Kn)$.
The first-order part of the distribution function $f\UP[1]$ can then be expressed as 
\begin{equation} \label{eq:expan_f1}
f\UP[1] = \sum_{i,j=1}^3 \frac{\osigma_{ij}}{2\theta^2} \phi_{ij}\UP[1] + \sum_{i=1}^3 \frac{2}{5} \frac{\oq_i}{\theta^2} \phi_i\UP[1].
\end{equation}

The evolution equations for $\osigma_{ij}$ and $\oq_j$ are derived by projecting the Boltzmann equation onto the basis functions $\phi_{ij}\UP[1]$ and $\phi_i\UP[1]$. Using the identity
\begin{equation}
    \label{eq:temporal_13}
    \left\lg\pd{f}{t}, \phi\right\rg=\pd{}{t}\left\lg f, \phi\right\rg-\left\lg f, \pd{\phi}{t}\right\rg,
\end{equation}
we obtain
\begin{equation}
    \label{eq:general_13}
    \begin{split}
        &\pd{\osigma_{ij}}{t}-\left\lg f, \pd{\phi_{ij}\UP[1]}{t}\right\rg+\left\lg \bxi\cdot \nabla_{\bx}f, \phi\UP[1]_{ij} \right\rg=\frac{1}{\Kn}\left\lg \mL[f], \phi\UP[1]_{ij}\right\rg+\frac{1}{\Kn}\left \lg \mQs[f,f], \phi\UP[1]_{ij}\right\rg, \\ 
        &\pd{\oq_i}{t}-\left\lg f, \pd{\phi_i\UP[1]}{t}\right\rg+\left\lg \bxi\cdot \nabla_{\bx}f, \phi\UP[1]_i \right\rg=\frac{1}{\Kn}\left\lg \mL[f], \phi\UP[1]_i\right\rg+\frac{1}{\Kn}\left \lg \mQs[f,f], \phi\UP[1]_i\right\rg.
    \end{split}
\end{equation}
The complete 13-moment system consists of these equations coupled with the conservation laws \eqref{eq:conserv_law}. To fully close the system, we must further specify $f\UP[k]$, ensuring that the model remains consistent with the variables used in the 13-moment formulation.

In the following subsections, we will introduce two closure methods for this model:
\begin{itemize}
\item Burnett-order closure ($n=1$ in \eqref{eq:closure}).
\item Super-Burnett-order closure ($n=2$ in \eqref{eq:closure2}\eqref{eq:fn}).
\end{itemize}

\begin{remark}
In Grad's 13-moment equations, the definitions of the stress tensor $\sigma_{ij}$ and the heat flux $q_i$ are independent of the collision kernel. In contrast, the variables $\osigma_{ij}$ and $\oq_i$ in our model depend on the linear part of the collision operator $\mL$ through the choice of basis functions \eqref{eq:V1basis}.
For Maxwell molecules, $\varphi_{ij}\UP[1]$ and $\varphi_i\UP[1]$ are eigenfunctions of $\mL$. Consequently, this leads to $\phi_{ij}\UP[1] = \varphi_{ij}\UP[1]$ and $\phi_i\UP[1] = \varphi_i\UP[1]$, making $\osigma_{ij}, \oq_i$ identical to $\sigma_{ij}, q_i$.
\end{remark}

\begin{remark}
Strictly speaking, $\osigma_{ij}$ and $\oq_i$ are not true ``moments'' of the distribution function, since $\vp_{ij}\UP[1]/\mM$ and $\vp_i\UP[1]/\mM$ are generally not polynomials. Nevertheless, we retain the \emph{13-moment model} terminology for simplicity. In fact, for many collision models (e.g., inverse power laws), the functions $\vp_{ij}\UP[1]/\mM$ and $\vp_i\UP[1]/\mM$ approximate polynomials quite well \cite{Cai2015, LinR13}. 
\end{remark}


\subsection{13-moment system of Burnett order}
\label{sec:R13_Burnett}
The Burnett order system is constructed to achieve an order of accuracy ($O(\Kn^2)$). According to Theorem \ref{thm:full}, it can be realized by truncating the expansion at $f\approx f\UP[0]+f\UP[1]$, meaning $f\UP[k] = 0$ for all $k \geqslant 2$. Under this approximation, the 13-moment system is derived by combining the conservation laws \eqref{eq:conserv_law} with the evolution equations \eqref{eq:general_13}, leading to:
\begin{equation}
    \label{eq:Burnett_13}
    \begin{split}
        &\pd{\rho}{t}+\bu\cdot\nabla_{\bx}\rho+\rho\nabla_{\bx}\cdot{\bu}=0, \\
        &\rho\pd{u_i}{t}+\rho\bu\cdot\nabla_{\bx}u_i+\theta\pd{\rho}{x_i}+\rho\pd{\theta}{x_i}+ 
        \left\lg\bxi\cdot\nabla_{\bx}f\UP[1], (\xi_i-u_i)\mM\right\rg=0, \\ 
        &\frac32\rho\pd{\theta}{t}+\frac32\rho\bu\cdot\nabla_{\bx}{\theta}+\rho\theta\nabla_{\bx}\cdot\bu+ 
        \left\lg\bxi\cdot\nabla_{\bx}f\UP[1], \left(\frac{|\bxi-\bu|^2}{2}-\frac{3\theta}2\right)\mM\right\rg=0, \\
        &\pd{\osigma_{ij}}{t}-\left\lg f\UP[0]+f\UP[1], \pd{\phi_{ij}\UP[1]}{t}\right\rg+\left\lg \bxi\cdot \nabla_{\bx}(f\UP[0]+f\UP[1]), \phi\UP[1]_{ij} \right\rg=\\
        &\qquad \qquad \qquad \qquad \qquad \frac{1}{\Kn}\left\lg \mL[f\UP[1]], \phi\UP[1]_{ij}\right\rg+\frac{1}{\Kn}\left \lg \mQ[f\UP[1],f\UP[1]], \phi\UP[1]_{ij}\right\rg, \\ 
        &\pd{\oq_i}{t}-\left\lg f\UP[0]+f\UP[1], \pd{\phi_i\UP[1]}{t}\right\rg+\left\lg \bxi\cdot \nabla_{\bx}(f\UP[0]+f\UP[1]), \phi\UP[1]_i \right\rg=\\
        &\qquad \qquad \qquad \qquad \qquad \frac{1}{\Kn}\left\lg \mL[f\UP[1]], \phi\UP[1]_i\right\rg+\frac{1}{\Kn}\left \lg \mQ[f\UP[1],f\UP[1]], \phi\UP[1]_i\right\rg,
    \end{split}
\end{equation}
where $f\UP[0]=\rho\mM$, and $f\UP[1]$ is given by \eqref{eq:expan_f1}. In these equations, the evolutions of $\osigma_{ij}$ and $\oq_i$
are handled following
\begin{align}
    & \left\lg f\UP[0] + f\UP[1], \pd{\phi_{ij}\UP[1]}{t} \right\rg = \left\lg f\UP[1], \pd{\phi_{ij}\UP[1]}{\theta} \right\rg (\mE_{\theta}\UP[0]+\mE_{\theta}\UP[1]) + \sum_{k=1} ^3 \left\lg f\UP[1], \pd{\phi_{ij}\UP[1]}{u_k} \right\rg (\mE_{u_k}\UP[0]+\mE_{u_k}\UP[1]), \label{eq:phi_ij} \\
    & \left\lg f\UP[0] + f\UP[1], \pd{\phi_i\UP[1]}{t} \right\rg = \left\lg f\UP[1], \pd{\phi_i\UP[1]}{\theta} \right\rg (\mE_{\theta}\UP[0]+\mE_{\theta}\UP[1]) + \sum_{k=1} ^3 \left\lg f\UP[1], \pd{\phi_i\UP[1]}{u_k} \right\rg (\mE_{u_k}\UP[0]+\mE_{u_k}\UP[1]), \label{eq:phi_i}
\end{align} 
where \eqref{eq:time_derivative}, \eqref{eq:re_conserv} and Proposition \ref{prop:Ker} are utilized. 
Substituting \eqref{eq:phi_ij} and \eqref{eq:phi_i} into \eqref{eq:Burnett_13} leads to a complete 13-moment system that only involves first-order derivatives.

For Maxwell molecules, the approximation $f \approx f\UP[0] + f\UP[1]$ aligns with the ansatz in Grad's moment method, making the resulting system \eqref{eq:Burnett_13} identical to Grad's 13-moment equations.
For general gas molecules, a different 13-moment system of Burnett order has been derived in \cite{Struchtrup2005R13} using the order of magnitude method. Although our derivation is also based on an analysis of magnitudes, the resulting equations differ due to the choice of variables. Specifically, our choice ensures that the system is linearly stable for any collision model, consistent with the stability analysis in \cite{LinR13}. The readers can also refer to \cite{Bunger2023} for more details.

\subsection{13-moment system with super-Burnett order}
\label{sec:R13_super_Burnett}
As mentioned at the end of Section \ref{sec:result}, the regularized model with $n=2$ attains super-Burnett accuracy ($O(\Kn^3)$) through the following system:
\begin{equation}
    \label{eq:super_Burnett_13}
    \begin{split}
        &\pd{\rho}{t}+\bu\cdot\nabla_{\bx}\rho+\rho\nabla_{\bx}\cdot{\bu}=0, \\
        &\rho\pd{u_i}{t}+\rho\bu\cdot\nabla_{\bx}u_i+\theta\pd{\rho}{x_i}+\rho\pd{\theta}{x_i}+ 
        \left\lg\bxi\cdot\nabla_{\bx}(f\UP[1]+f\UP[2]), (\xi_i-u_i)\mM\right\rg=0, \\ 
        &\frac32\rho\pd{\theta}{t}+\frac32\rho\bu\cdot\nabla_{\bx}{\theta}+\rho\theta\nabla_{\bx}\cdot\bu+ 
        \left\lg\bxi\cdot\nabla_{\bx}(f\UP[1]+f\UP[2]), \left(\frac{|\bxi-\bu|^2}{2}-\frac{3\theta}2\right)\mM\right\rg=0, \\
        &\pd{\osigma_{ij}}{t}-\left\lg f\UP[1], \pd{\phi_{ij}\UP[1]}{\theta}\right\rg\omE_{\theta}\UP[2]-\sum_{k=1}^3\left\lg f\UP[1], \pd{\phi_{ij}\UP[1]}{u_k}\right\rg\omE_{u_k}\UP[2]-\left\lg f\UP[1], \pd{\phi_{ij}\UP[2]}{\theta}\right\rg\omE_{\theta}\UP[1] \\
        &\quad -\sum_{k=1}^3\left\lg f\UP[2], \pd{\phi_{ij}\UP[1]}{u_k}\right\rg\omE_{u_k}\UP[1] +\left\lg \bxi\cdot \nabla_{\bx}(f\UP[0]+f\UP[1]+f\UP[2]), \phi\UP[1]_{ij} \right\rg=\\
        &\quad \frac{1}{\Kn}\left\lg \mL[f\UP[1]+f\UP[2]], \phi\UP[1]_{ij}\right\rg+\frac{1}{\Kn}\left \lg \mQ[f\UP[1],f\UP[1]]+\mQ[f\UP[1], f\UP[2]]+\mQ[f\UP[2],f\UP[1]], \phi\UP[1]_{ij}\right\rg, \\ 
        &\pd{\oq_i}{t}-\left\lg f\UP[1], \pd{\phi_{i}\UP[1]}{\theta}\right\rg\omE_{\theta}\UP[2]-\sum_{k=1}^3\left\lg f\UP[1], \pd{\phi_{i}\UP[1]}{u_k}\right\rg\omE_{u_k}\UP[2]-\left\lg f\UP[1], \pd{\phi_{i}\UP[2]}{\theta}\right\rg\omE_{\theta}\UP[1] \\
        &\quad -\sum_{k=1}^3\left\lg f\UP[2], \pd{\phi_{i}\UP[1]}{u_k}\right\rg\omE_{u_k}\UP[1]+\left\lg \bxi\cdot \nabla_{\bx}(f\UP[0]+f\UP[1]+f\UP[2]), \phi\UP[1]_i \right\rg=\\
        &\quad \frac{1}{\Kn}\left\lg \mL[f\UP[1]+f\UP[2]], \phi\UP[1]_i\right\rg+\frac{1}{\Kn}\left \lg \mQ[f\UP[1],f\UP[1]]+\mQ[f\UP[1], f\UP[2]]+\mQ[f\UP[2],f\UP[1]], \phi\UP[1]_i\right\rg,
    \end{split}
\end{equation}
where $f\UP[0]=\rho\mM$ and $f\UP[1]$ is determined by \eqref{eq:expan_f1}. The notation $\omE$ is defined as
\begin{equation}
    \label{eq:def_omE}
    \omE_{\theta}\UP[l]=\sum_{j=0}^l \mE_{\theta}\UP[j], \qquad \omE_{u_k}\UP[l]=\sum_{j=0}^l \mE_{u_k}\UP[j].
\end{equation}
%
The second-order term $f\UP[2]$ is introduced to achieve the order of accuracy $O(\Kn^3)$ and is expressed as:
\begin{equation} \label{eq:f2}
\begin{split}
    f\UP[2] &= \Kn (\mL\UP[2])^{-1} \mP_{(\bbV\UP[1])^{\perp} \cap \bbV\UP[2]} \Bigg[ \omE\UP[1]_{\theta}\left(\sum_{i,j=1}^3 \frac{\osigma_{ij}}{2\theta^2} \pd{\phi_{ij}\UP[1]}{\theta} + 
    \sum_{i=1}^3 \frac{2}{5} \frac{\oq_i}{\theta^2} \pd{\phi_i\UP[1]}{\theta}\right)+ \\
    & \sum_{k=1}^3\omE_{u_k}\UP[1]\left(\sum_{i,j}^3 \frac{\osigma_{ij}}{2\theta^2} \pd{\phi_{ij}\UP[1]}{u_k} + \sum_{i=1}^3 \frac{2}{5} \frac{\oq_i}{\theta^2} \pd{\phi_i\UP[1]}{u_k}\right)
    +\bxi\cdot\nabla_{\bx}(f\UP[0]+f\UP[1]) - \frac{1}{\Kn} \mQ[f\UP[1], f\UP[1]]\Bigg],
\end{split}
\end{equation}
where $\mL\UP[2]$ is an operator on $(\bbV\UP[1])^{\perp} \cap \bbV\UP[2]$ satisfying $\mL\UP[2] g = \mP_{(\bbV\UP[1])^{\perp} \cap \bbV\UP[2]} \mL g$ for any $g \in (\bbV\UP[1])^{\perp} \cap \bbV\UP[2]$. The invertibility of $\mL\UP[2]$ is shown in Proposition \ref{prop:invert}. Notably, the equation \eqref{eq:f2} is equivalent to \eqref{eq:fn} when $n=2$.

This expression of $f\UP[2]$ contains at most first-order spatial derivatives. Substituting \eqref{eq:f2} into \eqref{eq:super_Burnett_13} yields a model containing only second-order derivatives.
Comparing with the standard super-Burnett equations \cite{Struchtrup} which include fourth-order derivatives, this regularized 13-moment model clearly has a much neater form.

The precise expressions requires characterizing the subspace $\bbV\UP[2]$, which depends on the collision model. 
For Maxwell molecules, $\dim \bbV\UP[2] = 26$, and the equations \eqref{eq:super_Burnett_13} are identical to the regularized 13-moment equations in \cite{Struchtrup}. However, for general collision operators, the subspace $\bbV\UP[2]$ may include more than $80$ dimensions, which makes the system very complicated.
Detailed studies will be presented in our future work.
\section{Proofs of theorems}
\label{sec:validation}
This section provides mathematical proofs of the theorems stated in Section \ref{sec:result}. We first validate the properties of the function spaces $\bbV\UP[k]$ (Theorem \ref{thm:order}) in Section \ref{sec:proof_thm_order}, then prove the order of accuracy for our reduced models (Theorems \ref{thm:full} and \ref{thm:linear}) in Section \ref{sec:proof_model_order}.

\subsection{Proof of Theorem \ref{thm:order}}
\label{sec:proof_thm_order}
\begin{proof}[Proof of Theorem \ref{thm:order}]
    The proof proceeds by induction. 
    
    We first consider statement \ref{state:finite}. For the base case $k=0$, the definition \eqref{eq:def_V0}\eqref{eq:def_F0} shows:
    \begin{itemize}
    \item $\bbW\UP[0]$ contains 1 element.
    \item $\bbF\UP[0]$ contains 13 elements.
    \item $\bbV\UP[0]$ consists of 5 dimensions.
    \end{itemize}
    Now we assume that $\bbW\UP[0], \cdots, \bbW\UP[k-1]$ and $\bbF\UP[0], \cdots, \bbW\UP[k-1]$ are all finite. By definition \eqref{eq:def_VF}, the finite nature of $\bbF\UP[k-1]$ and the bijectivity of $\mL$ on $(\bbV\UP[0])^{\perp}$ imply that $\bbW\UP[k-1]$ is also finite, resulting in the finite dimensionality of $\bbV\UP[k]$. Furthermore, $\bbF\UP[k]$ is also finite because $\bbW\UP[0], \bbW\UP[1], \cdots, \bbW\UP[k]$ all contain finite elements. By induction, statement \ref{state:finite} is validated for all $k$.

    We now address statement \ref{state:order}. When $k=0$, we have $f=f_{\rm eq}+g$ with $g\in\left(\bbV\UP[0]\right)^{\perp}$ by \eqref{eq:Maxwellian}, and the Chapman-Enskog expansion in Section \ref{sec:CE} implies $g = O(\Kn)$.
    For any $\psi\UP[1] \in \left(\bbV\UP[0]\right)^{\perp}$, it can be noticed that
    \begin{equation}
        \label{eq:proof_k0_p2}
        \lg \psi\UP[1], f \rg=\lg \psi\UP[1], f_{\rm eq}+g \rg=\lg\psi\UP[1], g \rg = O(\Kn),
    \end{equation}
    which verifies the base case $k=0$.

    Suppose for all $l = 0, 1, \cdots, k-1$, it holds that $\lg \psi\UP[l+1], f \rg = O(\Kn^{l+1})$ for any $\psi\UP[l+1] \in (\bbV\UP[l])^{\perp}$, and our objective is to show that $\lg \psi\UP[k+1], f \rg = O(\Kn^{k+1})$ for all $\psi\UP[k+1] \in (\bbV\UP[k])^{\perp}$.
    We decompose $f$ into $k+1$ terms:
    \begin{equation}
        \label{eq:exp_f_p1}
        f=\sum_{l=0}^{k-1}f\UP[l]+f^R,
    \end{equation}
    where
    \begin{equation}
        \label{eq:exp_f_p2}
        f\UP[0]=\feq=\rho \mM\in\bbV\UP[0],\quad f\UP[l] = \mP_{\bbV\UP[l]} f - \mP_{\bbV\UP[l-1]} f \in\left(\bbV\UP[l-1]\right)^{\perp}\cap \bbV\UP[l], \quad 1\leqslant l \leqslant k-1,
    \end{equation}
    which yields $f^R\in\left(\bbV\UP[k-1]\right)^{\perp}$. By the induction hypothesis, we know $f\UP[l]=O(\Kn^l)$ for $l\leqslant k-1$ and $f^R = O(\Kn^k)$.

    Since $\psi\UP[k+1]\in \left(\bbV\UP[k]\right)^{\perp}\subset \left(\bbV\UP[0]\right)^{\perp}$, it follows that $\psi\UP[k+1] = \mL \mLd \psi\UP[k+1]$. Consequently,
    \begin{equation}
        \label{eq:thm1_proof_p1}
        \begin{split}
            \left\lg f, \psi\UP[k+1]\right\rg&=\left\lg f, \mL\mLd\psi\UP[k+1]\right\rg=
            \left\lg \mLd (\mL f), \psi\UP[k+1]\right\rg \\
            &=\left\lg \mLd\left(\Kn\pd{f}{t}+\Kn\bxi\cdot \nabla_{\bx}f- \mQs[f,f]\right),  \psi\UP[k+1] \right\rg,
        \end{split}
    \end{equation}
    where we have used the symmetry of $\mL$ and substituted the Boltzmann equation \eqref{eq:re_Boltz} for $\mL f$. 
    
    To proceed, we decompose \eqref{eq:thm1_proof_p1} into two terms as follows:
    $$\Kn\pd{f}{t}+\Kn\bxi\cdot \nabla_{\bx}f- \mQs[f,f]=F\UP[k-1]+G\UP[k-1],$$ 
    where the principal part is given by
    \begin{equation}
        \label{eq:Fk}
        F\UP[k-1]= \Kn\sum_{l=0}^{k-1}\pd{f\UP[l]}{t}+\Kn\sum_{l=0}^{k-1}\bxi\cdot\nabla_{\bx}f\UP[l]-\sum_{\substack{r,s=1,\cdots,k-1\\ r+s\leqslant k}}\mQ[f\UP[r],f\UP[s]],
    \end{equation}
    and the remainder $G\UP[k-1]$ has order $O(\Kn^{k+1})$ according to the decomposition \eqref{eq:exp_f_p1}.
    
    We now demonstrate that $\mLd F\UP[k-1]$ is orthogonal to $\psi\UP[k+1]$, thereby confirming the desired magnitude of \eqref{eq:thm1_proof_p1}.
    Since $f\UP[l] \in \bbV\UP[l]$ and $\mP_{\bbV\UP[0]}f\UP[l] \in \rspan{\mM}$, it holds $f\UP[l]\in\rspan(\bbW\UP[l])$, allowing us to represent $f\UP[l]$ as a linear combination of the functions in $\bbW\UP[l]$:
    \begin{equation}
        \label{eq:expan_fl}
        f\UP[l]=\sum_{\alpha=1}^{N\UP[l]} w\UP[l]_{\alpha}\vp_{\alpha}\UP[l],
    \end{equation}
    where $\{\vp_{\alpha}\UP[l]\mid \alpha = 1,\cdots, N\UP[l]\} = \bbW\UP[l]$.
    The terms in \eqref{eq:Fk} are thus written as the following summations:
    \begin{equation}
        \begin{aligned}
            \pd{f\UP[l]}{t}&=\sum_{\alpha=1}^{N\UP[l]} \left(\pd{w\UP[l]_{\alpha}}{t}\vp_{\alpha}\UP[l]+
            w\UP[l]_{\alpha}\sum_{i=1}^3\pd{u_i}{t}\pd{\vp_{\alpha}\UP[l]}{u_i}+
            w\UP[l]_{\alpha}\pd{\theta}{t}\pd{\vp_{\alpha}\UP[l]}{\theta}\right), \\
            \bxi \cdot \nabla_{\bx} f\UP[l] &= \sum_{\alpha=1}^{N\UP[l]} \sum_{j=1}^3 \left(\pd{w\UP[l]_{\alpha}}{x_j}\xi_j\vp_{\alpha}\UP[l]+
            w\UP[l]_{\alpha}\sum_{i=1}^3\pd{u_i}{x_j}\xi_j\pd{\vp_{\alpha}\UP[l]}{u_i}+w\UP[l]_{\alpha}
            \pd{\theta}{x_j}\xi_j\pd{\vp_{\alpha}\UP[l]}{\theta}\right), \\
            \mQ^*[f\UP[r], f\UP[s]] &= \sum_{\alpha=1}^{N\UP[r]} \sum_{\beta=1}^{N\UP[s]} w_{\alpha}\UP[r] w_{\beta}\UP[s] \mQ^*[\varphi_{\alpha}\UP[r], \varphi_{\beta}\UP[s]].
        \end{aligned}
    \end{equation}
    Due to the nesting property of $\bbW\UP[k]$, \textit{i.e.} $\bbW\UP[0] \subset \bbW\UP[1] \subset \cdots \subset \bbW\UP[k-1]$, all these terms are linear combinations of functions in $\bbF\UP[k-1]$ defined in \eqref{eq:def_VF}. This ensures that $F\UP[k-1] \in \rspan (\bbF\UP[k-1])$, and consequently $\mLd F\UP[k-1] \in \bbV\UP[k]$.
    Thus, given $\psi\UP[k+1] \perp \bbV\UP[k]$, the magnitude of \eqref{eq:thm1_proof_p1} can obtained by
    \begin{equation}
        \label{eq:thm1_proof_p2}
        \begin{split}
        \lg f, \psi\UP[k+1] \rg &= \lg \mLd F\UP[k-1], \psi\UP[k+1] \rg + \lg \mLd G\UP[k-1], \psi\UP[k+1] \rg \\
        &= \lg \mLd G\UP[k-1], \psi\UP[k+1] \rg = O(\Kn^{k+1}),
        \end{split}
    \end{equation}
    which completes the induction. 
\end{proof}

\subsection{Proofs of Theorems \ref{thm:full} and \ref{thm:linear}}
\label{sec:proof_model_order}
In Section \ref{sec:thm}, the orders of accuracy for our reduced models are summarized in Theorems \ref{thm:full} and \ref{thm:linear}. Before presenting the proofs, we provide several preliminary results.

\subsubsection{Some lemmas}
\label{sec:lemmas}
For simplicity, we first define a linear operator $\mD$ as
\begin{equation}
    \label{eq:def_mD}
    \mD f=\pd{f}{t}+{\bxi}\cdot\nabla_{\bx}f.
\end{equation}
With this definition, the Boltzmann equation \eqref{eq:re_Boltz} can be reformulated as
\begin{displaymath}
\mD f = \frac{1}{\Kn} \mL f + \frac{1}{\Kn} \mQ^*[f,f],
\end{displaymath}
and the Maxwellian iteration \eqref{eq:re_Maxwell_iter} is rewritten as
\begin{equation}
    \label{eq:origin_Maxwell}
    f_l = f\UP[0] + \mLd \left( \Kn \mD f_{l-1} - \mQs_l \right).
\end{equation}

The hyperbolic reduced model \eqref{eq:closure} can be expressed in operator form for a function $\hf\in\bbV\UP[n]$:
\begin{equation}
    \label{eq:oper_moment}
    \hmD \hat{f} = \frac{1}{\Kn} \hmL \hat{f} + \frac{1}{\Kn}\hmQs[\hf,\hf], 
\end{equation}
where $\hmD = \mP_{\bbV\UP[n]}\mD\mP_{\bbV\UP[n]}$, $\hmL = \mP_{\bbV\UP[n]}\mL\mP_{\bbV\UP[n]}$, and $\hmQs$ represents the truncated collision operator:
\begin{equation}
    \label{eq:def_hmQs}
    \hmQs[\hf,\hf]=\sum_{\substack{i,j \geqslant 1\\ i+j \leqslant n+1}} \mP_{\bbV\UP[n]}\mQ[(\mP_{\bbV\UP[i]} - \mP_{\bbV\UP[i-1]}) \hat{f},(\mP_{\bbV\UP[j]} - \mP_{\bbV\UP[j-1]}) \hat{f}].
\end{equation}
Correspondingly, the model's Maxwellian iteration takes the form:
\begin{equation}
\label{eq:moment_Maxwell}
\hat{f}_l = f\UP[0] + \hmLd \left( \Kn \hmD \hat{f}_{l-1} - \hmQs_l \right),
\end{equation}
where 
\begin{equation}
    \label{eq:def_hmQs_l}
    \hmQs_l=\sum_{\substack{r,s,i,j \geqslant 1\\ r+s \leqslant l\\ i+j \leqslant n+1}} \mP_{\bbV\UP[n]}\mQ[(\mP_{\bbV\UP[i]} - \mP_{\bbV\UP[i-1]})(\hat{f}_r - \hat{f}_{r-1}), (\mP_{\bbV\UP[j]} - \mP_{\bbV\UP[j-1]})(\hat{f}_s - \hat{f}_{s-1})].
\end{equation}

Notably, $\hat{f}_l \in \bbV\UP[n]$ for all $l$. 
According to the definition \ref{def:OOA}, the reduced model \eqref{eq:oper_moment} has order of accuracy $O(\Kn^k)$ if and only if $\mP_{\bbV\UP[0]}\mD f_k=\mP_{\bbV\UP[0]}\hmD\hat{f}_k+O(\Kn^{k+1})$.

Similarly, the regularized reduced model \eqref{eq:closure2}\eqref{eq:fn} and its corresponding Maxwellian iteration can be reformulated in operator form as follows:
\begin{align}
    \label{eq:oper_regularized}
    \wmD\wf &= \frac{1}{\Kn}\wmL\wf + \frac{1}{\Kn}\wmQs[\wf, \wf], \\
    \label{eq:regular_Maxwell}
    \wf_l &= f\UP[0] + \wmLd \left( \Kn \wmD \wf_{l-1} - \wmQs_l \right),
\end{align}
where $\wmL = \hmL$ and
\begin{equation}
    \label{eq:def_wmQs_l}
    \begin{split}
        \wmQs_l=&\sum_{\substack{r,s,i,j \geqslant 1\\ r+s \leqslant l\\ i+j \leqslant n+1}} \mP_{\bbV\UP[n-1]}\mQ[(\mP_{\bbV\UP[i]} - \mP_{\bbV\UP[i-1]})(\wf_r - \wf_{r-1}), (\mP_{\bbV\UP[j]} - \mP_{\bbV\UP[j-1]})(\wf_s - \wf_{s-1})] \\
        + & \sum_{\substack{r,s,i,j \geqslant 1\\ r+s \leqslant l\\ i+j \leqslant n}} \mP_{(\bbV\UP[n-1])^{\perp}\cap\bbV\UP[n]} \mQ[(\mP_{\bbV\UP[i]} - \mP_{\bbV\UP[i-1]})(\wf_r - \wf_{r-1}), (\mP_{\bbV\UP[j]} - \mP_{\bbV\UP[j-1]})(\wf_s - \wf_{s-1})].
    \end{split}
\end{equation}
The expansion of \eqref{eq:oper_regularized} using the basis functions $\varphi_{\alpha}\UP[l]$ is consistent with \eqref{eq:closure2}\eqref{eq:fn}, and the detailed definition of $\wmD$ is omitted due to its complicated form.

A direct comparison between systems \eqref{eq:closure} and \eqref{eq:closure2}\eqref{eq:fn} reveals the following properties:
\begin{enumerate}[label=\textbf{(P\arabic*)}]
    \item For any $\phi\UP[n]\in\bbV\UP[n]$, it holds $\mP_{\bbV\UP[n-1]}\wmD \phi\UP[n]= \mP_{\bbV\UP[n-1]}\hmD \phi\UP[n] + O(\Kn^{2n})$; \\for any $\phi\UP[n-1]\in\bbV\UP[n-1]$, it holds $\mP_{\bbV\UP[n]}\wmD \phi\UP[n-1]= \mP_{\bbV\UP[n]}\hmD \phi\UP[n-1] $; \label{prop:wmD}
    \item For any distributions $f,g$, it holds $\mP_{\bbV\UP[n-1]}\wmQs[f,g]=\mP_{\bbV\UP[n-1]}\hmQs[f,g]$ and $\mP_{\bbV\UP[n-1]}\wmQs[f,g] = \mP_{\bbV\UP[n-1]}\hmQs[f,g]$. \label{prop:wmQ}
\end{enumerate}
These properties can be readily verified by analyzing the differences between the two systems with $f\UP[k]=O(\Kn^k)$ in \eqref{eq:def_mE}.

Now, we present some lemmas for the operators in \eqref{eq:oper_moment}:
\begin{lemma}
    \label{lemma:DL_hess}
    For any $k\geqslant 1$, For any $k\geqslant 1$, the following identities hold:
    \begin{equation}
    \label{eq:DL_hess}
    \mP_{\bbV\UP[k-1]}\mD \mLd\mP_{(\bbV\UP[k])^{\perp}}=0, \qquad \mP_{(\bbV\UP[k])^{\perp}}\mLd\mD\mP_{\bbV\UP[k-1]}=0.
    \end{equation} 
\end{lemma}
\begin{proof}[Proof of Lemma \ref{lemma:DL_hess}]
    It suffices to show that for any $\vp\UP[k-1]\in\bbV\UP[k-1]$ and $\vp^R\perp \bbV\UP[k]$, the following equalities are satisfied:
    \begin{equation}
        \label{eq:lemma_hess_1}
        \left\lg\vp\UP[k-1], \mD\mLd\vp^R\right\rg=0, \qquad
        \left\lg\vp^R, \mLd\mD\vp\UP[k-1]\right\rg=0.
    \end{equation}
    For the first equation in \eqref{eq:lemma_hess_1}, with direct computations, we have
    \begin{equation}
        \label{eq:D_time}
        \begin{split}
            \left\lg\vp\UP[k-1], \mD\mLd\vp^R\right\rg
            &=\left\lg\vp\UP[k-1], \pd{}{t}\mLd\vp^R\right\rg+\left\lg\vp\UP[k-1], \nabla_{\bx}\cdot(\bxi\mLd\vp^R)\right\rg \\
            &=\pd{}{t}\left\lg\vp\UP[k-1], \mLd\vp^R\right\rg-\left\lg\pd{}{t}\vp\UP[k-1], \mLd\vp^R\right\rg \\
            &\quad +\nabla_{\bx}\cdot \left\lg\bxi\vp\UP[k-1], \mLd\vp^R\right\rg-\left\lg\nabla_{\bx}\cdot (\bxi\vp\UP[k-1]), \mLd\vp^R\right\rg \\
            &=\pd{}{t}\left\lg\mLd\vp\UP[k-1], \vp^R\right\rg-\left\lg\mLd\pd{}{t}\vp\UP[k-1], \vp^R\right\rg \\
            &\quad +\nabla_{\bx}\cdot \left\lg\mLd(\bxi\vp\UP[k-1]), \vp^R\right\rg-\left\lg\mLd\Big(\nabla_{\bx}\cdot (\bxi\vp\UP[k-1])\Big), \vp^R\right\rg.
        \end{split}
    \end{equation}
    
    Following the approach used in the proof of Theorem \ref{thm:order}, we can derive 
    $$\vp\UP[k-1], \pd{}{t}\vp\UP[k-1], \bxi\vp\UP[k-1], \nabla_{\bx}\cdot (\bxi\vp\UP[k-1])\in \rspan(\bbF\UP[k-1]),$$ 
    which yields 
    $$
    \mLd\vp\UP[k-1],\, \mLd\pd{}{t}\vp\UP[k-1],\, \mLd(\bxi\vp\UP[k-1]),\, \mLd\Big(\nabla_{\bx}\cdot (\bxi\vp\UP[k-1])\Big)\in \bbV\UP[k].
    $$
    The proof is completed using $\vp^R \perp \bbV\UP[k]$.
    
    The second equation in \eqref{eq:lemma_hess_1} can be validated using the fact that $\mD\vp\UP[k-1]\in\rspan(\bbF\UP[k-1])$ and $\mLd\mD\vp\UP[k-1]\in\bbV\UP[k]$.
\end{proof}

\begin{proposition}
    \label{prop:invert}
    For a linear space $\bbA \in (\bbV\UP[0])^{\perp}$ with finite dimensions, the operator $\mP_{\bbA} \mL \mP_{\bbA}: \;\bbA\to \bbA$ is invertible.
\end{proposition}
\begin{proof}
    Since $\mP_{\bbA} \mL \mP_{\bbA}$ is a linear operator on a finite-dimensional space $\bbA$, it is sufficient to demonstrate that
    ${\rm Ker} (\mP_{\bbA} \mL \mP_{\bbA})=\{0\}$.
    
    Assume, for contradiction, that there exists a nonzero element $\phi \in \bbA$ such that $\mP_{\bbA} \mL \mP_{\bbA} \phi = 0$.
    Given that $\phi \in \bbA$ and $\bbA \subset (\bbV\UP[0])^{\perp}$, we directly have $\lg \mL\phi, \phi \rg<0 $ from Hypothesis \ref{Hypo:semi-defi}, which contradicts the assumption. 
\end{proof}

\begin{lemma}
    \label{lemma:DL_deter}
    For all $0\leqslant i\leqslant n-1$ and $0\leqslant j\leqslant n$, the following relations hold: 
    \begin{equation}
        \label{eq:DL_deter}
        \mP_{\bbV\UP[i]}\mD \mLd\mP_{\bbV\UP[j]}=\mP_{\bbV\UP[i]}\hmD \hmLd\mP_{\bbV\UP[j]},\qquad\mP_{\bbV\UP[j]}\mLd \mD \mP_{\bbV\UP[i]} = \mP_{\bbV\UP[j]}\hmLd \hmD \mP_{\bbV\UP[i]}.
    \end{equation}
\end{lemma}

\begin{proof}[Proof of Lemma \ref{lemma:DL_deter}]
We first note that $\bbV\UP[i] \subset \bbV\UP[n-1]$ and $\bbV\UP[j] \subset \bbV\UP[n]$, so it suffices to verify $\mP_{\bbV\UP[n-1]}\mD \mLd\mP_{\bbV\UP[n]}=\mP_{\bbV\UP[n-1]}\hmD \hmLd\mP_{\bbV\UP[n]}$ and $\mP_{\bbV\UP[n]}\mLd \mD \mP_{\bbV\UP[n-1]} = \mP_{\bbV\UP[n]}\hmLd \hmD \mP_{\bbV\UP[n-1]}$.
We will focus on proving the first equality, as the second can be demonstrated analogously. 

Let $\mathbb{S} = (\bbV\UP[0])^{\perp} \cap \bbV\UP[n]$ and define
\begin{gather*}
    \mD' = \mP_{\bbV\UP[n-1]} \mD \mP_{\mathbb{S}}, \quad
    \mK' = \mP_{\bbV\UP[n-1]} \mD \mLd \mP_{\mathbb{S}}, \quad
    \mL' = \mP_{\mathbb{S}} \mL \mP_{\mathbb{S}}, \\
    \hmD' = \mP_{\bbV\UP[n-1]} \hmD \mP_{\mathbb{S}}, \quad
    \hmK' = \mP_{\bbV\UP[n-1]} \hmD \hmLd \mP_{\mathbb{S}}, \quad
    \hmL' = \mP_{\mathbb{S}} \hmL \mP_{\mathbb{S}}.
\end{gather*}
Given that $\mLd \mP_{\bbV\UP[0]} = \hmLd \mP_{\bbV\UP[0]} = 0$, it is clear that $\mK' = \mP_{\bbV\UP[n-1]} \mD \mLd \mP_{\bbV\UP[n]}$ and $\hmK' = \mP_{\bbV\UP[n-1]} \hmD \mLd \mP_{\bbV\UP[n]}$. Thus, our objective reduces to showing $\mK' = \hmK'$.

For clarity, let $\mathcal{O}\big|{\mathbb{A}}$ denote the restriction of an operator $\mathcal{O}$ to a subspace $\mathbb{A}$. Notably, we have $\mK|_{\mathbb{S}^{\perp}} = \mK'|_{\mathbb{S}^{\perp}} = 0$ due to the projection $\mP_{\mathbb{S}}$ in their definitions. 
Therefore, it remains only to prove $\mK'|_{\mathbb{S}} = \hmK'|_{\mathbb{S}}$.

The three operators $\mD'$, $\mK'$ and $\mL'$ satisfy
\begin{equation} \label{eq:DKL}
    \mD' = \mK' \mL',
\end{equation}
which can be verified as follows:
\begin{displaymath}
    \begin{aligned}
        \mP_{\bbV\UP[n-1]} \mD \mP_{\mathbb{S}} &= \mP_{\bbV\UP[n-1]} \mD \mP_{(\bbV\UP[0])^{\perp}} \mP_{\mathbb{S}} = \mP_{\bbV\UP[n-1]} \mD \mLd \mL \mP_{\mathbb{S}} \\
        &= \mP_{\bbV\UP[n-1]} \mD \mLd \left(\mP_{\bbV\UP[0]} + \mP_{\mathbb{S}} + \mP_{(\bbV\UP[n])^{\perp}} \right) \mL \mP_{\mathbb{S}} = \mP_{\bbV\UP[n-1]} \mD \mLd  \mP_{\mathbb{S}} \mL \mP_{\mathbb{S}}
    \end{aligned}
\end{displaymath}
Here, we have utilized Lemma \ref{lemma:DL_hess} and the conservation property $\mLd \mP_{\bbV\UP[0]} = 0$ in the second line.
By Proposition \ref{prop:invert}, the operator $\mL'$ is invertible on $\mathbb{S}$. Consequently, restricting the operators $\mD'$, $\mK'$, and $\mL'$ to $\mathbb{S}$ yields:
\begin{displaymath}
    \mK'|_{\mathbb{S}} = \mD'|_{\mathbb{S}} (\mL'|_{\mathbb{S}})^{-1},
\end{displaymath}
where $\mL'|_{\mathbb{S}}$ is interpreted as a mapping of $\mathbb{S}\longrightarrow \mathbb{S}$.
Using a similar argument, we can also establish $\hmK'|_{\mathbb{S}} = \hmD'|_{\mathbb{S}} (\hmL'|_{\mathbb{S}})^{-1}$. Given that $\mD' = \hmD'$ and $\mL' = \hmL'$, we immediately obtain the desired conclusion $\mK'|_{\mathbb{S}} = \hmK'|_{\mathbb{S}}$.
\end{proof}

\begin{corollary}
    \label{coro:hat_DL_hess}
    For $1\leqslant k\leqslant n$, it holds 
    $$\mP_{\bbV\UP[k-1]}\hmD \hmLd\mP_{(\bbV\UP[k])^{\perp}}=0,\qquad \mP_{(\bbV\UP[k])^{\perp}}\hmLd\hmD\mP_{\bbV\UP[k-1]}=0.$$ 
\end{corollary}
\begin{proof}
    The support and range of $\hmL$ are both contained within $\bbV\UP[n]$, which implies $\hmLd=\hmLd \mP_{\bbV\UP[n]}=\mP_{\bbV\UP[n]}\hmLd$. With direct computations, we have      
    \begin{equation}
        \label{eq:hat_DL_comp}
        \begin{split}
            \mP_{\bbV\UP[k-1]}\hmD \hmLd\mP_{(\bbV\UP[k])^{\perp}}&=\mP_{\bbV\UP[k-1]}\hmD \hmLd\mP_{\bbV\UP[n]}\mP_{(\bbV\UP[k])^{\perp}} = \mP_{\bbV\UP[k-1]}\mD \mLd\mP_{\bbV\UP[n]}\mP_{(\bbV\UP[k])^{\perp}} \\ 
            &= \mP_{\bbV\UP[k-1]}\mD \mLd\mP_{(\bbV\UP[k])^{\perp}}\mP_{\bbV\UP[n]}=0.
        \end{split}
    \end{equation}
    The calculation uses the fact that $\mP_{(\bbV\UP[k])^{\perp}}$ and $\mP_{\bbV\UP[n]}$ are commutable, which can be validated as follows:
    \begin{equation}
        \label{eq:commut}
        \mP_{\bbV\UP[n]}\mP_{(\bbV\UP[k])^{\perp}}=\mP_{\bbV\UP[n]}(\mI-\mP_{\bbV\UP[k]})=\mP_{\bbV\UP[n]}-\mP_{\bbV\UP[k]}=(\mI-\mP_{\bbV\UP[k]})\mP_{\bbV\UP[n]}=\mP_{(\bbV\UP[k])^{\perp}}\mP_{\bbV\UP[n]}.
    \end{equation}
    An identical approach can be used to establish that $\mP_{(\bbV\UP[k])^{\perp}}\hmLd\hmD\mP_{\bbV\UP[k-1]}=0$.
\end{proof}

The above statements will be applied to the analysis of the linearized case $\mQs = 0$ (Theorem \ref{thm:linear}). To complete the proof of Theorem \ref{thm:full}, we now present a few additional results concerning quadratic operators in the following lemmas.
\begin{lemma}
    \label{lemma:hmQs}
    For $l\leqslant k$, it holds that
    \begin{equation}
        \label{eq:lemma_hmQs}
        \sum_{\substack{r,s,i,j \geqslant 1\\ r+s \leqslant l\\ i+j \leqslant k}} \mQ[(\mP_{\bbV\UP[i]} - \mP_{\bbV\UP[i-1]})(f_r - f_{r-1}), (\mP_{\bbV\UP[j]} - \mP_{\bbV\UP[j-1]})(f_s - f_{s-1})] = \mQs_l,
    \end{equation}
     where $\mQs_l$ is defined in \eqref{eq:def_mQs_l}.
\end{lemma}
\begin{proof}[Proof of Lemma \ref{lemma:hmQs}]
    Since $f_s\in \bbV\UP[s]$, when $i>r$ or $j>s$, any term satisfying $i > r$ or $j > s$
    $$
    \mQ[(\mP_{\bbV\UP[i]} - \mP_{\bbV\UP[i-1]})(f_r - f_{r-1}), (\mP_{\bbV\UP[j]} - \mP_{\bbV\UP[j-1]})(f_s - f_{s-1})]
    $$
    vanishes. This allows us to restrict the summation over $i \leqslant r$ and $j \leqslant s$. When $l\leqslant k$, Additionally, for $l \leqslant k$, the condition $i + j \leqslant k$ becomes redundant. As a result, the left-hand side of \eqref{eq:lemma_hmQs} simplifies to
    \begin{equation}
        \label{eq:simp_hmQs}
        \begin{split}
            {\rm LHS}&=\sum_{\substack{r,s \geqslant 1\\ r+s \leqslant l}}\sum_{i=1}^r\sum_{j=1}^s \mQ[(\mP_{\bbV\UP[i]} - \mP_{\bbV\UP[i-1]})(f_r - f_{r-1}), (\mP_{\bbV\UP[j]} - \mP_{\bbV\UP[j-1]})(f_s - f_{s-1})] \\
            &=\sum_{\substack{r,s \geqslant 1\\ r+s \leqslant l}}\mQ[\mP_{\bbV\UP[r]}(f_r - f_{r-1}), \mP_{\bbV\UP[s]}(f_s - f_{s-1})]
            =\sum_{\substack{r,s \geqslant 1\\ r+s \leqslant l}}\mQ[f_r - f_{r-1}, f_s - f_{s-1}] \\
            &=\mQs_l.
        \end{split}
    \end{equation}
\end{proof}

\begin{lemma}
    \label{lemma:Pn_mQs}
    For $l\leqslant n$, it holds 
    \begin{equation}
        \label{eq:Pn_mQs}
        \mP_{\bbV\UP[n]}\mLd \mQs_l=\mP_{\bbV\UP[n]}\hmLd\mP_{\bbV\UP[n]}\mQs_l.
    \end{equation}
\end{lemma}
\begin{proof}[Proof of Lemma \ref{lemma:Pn_mQs}]
    Denote $\mR_l=\mLd \mQs_l$. Then $\mR_l\in (\bbV\UP[0])^{\perp}\cap\bbV\UP[l]$, and the left-hand side of \eqref{eq:Pn_mQs} can be transformed as follows:
    \begin{displaymath}
        \begin{aligned}
            \mP_{\bbV\UP[n]} \hmLd \mP_{\bbV\UP[n]}\mQs_l&= \mP_{\bbV\UP[n]} \hmLd \mP_{\bbV\UP[n]}\mP_{(\bbV\UP[0])^{\perp}}\mQs_l = \mP_{\bbV\UP[n]}\hmLd\mP_{\bbV\UP[n]}\mL \mLd \mQs_l \\
            &= \mP_{\bbV\UP[n]}\hmLd\mP_{\bbV\UP[n]}\mL \mR_l = \mP_{\bbV\UP[n]}\hmLd\mP_{\bbV\UP[n]}\mL \mP_{\bbV\UP[n]}\mR_l =\mP_{\bbV\UP[n]}\hmLd\hmL \mR_l \\
            &= \mP_{\bbV\UP[n]}\mR_l =\mP_{\bbV\UP[n]} \mLd \mQs_l.
        \end{aligned}
    \end{displaymath}
\end{proof}

In the following, proofs of Theorems \ref{thm:full} and \ref{thm:linear} will be presented.

\subsubsection{Proof for the reduced model with linearized collision operators}
\label{sec:lin_order}
We first consider the analysis of the reduced models for the Boltzmann equation with a linearized collision operator, as described in Theorem \ref{thm:linear}. In this context, the quadratic terms $\mbQ\sUP[k]$ are set to zero, and the Maxwellian iteration schemes \eqref{eq:origin_Maxwell}, \eqref{eq:moment_Maxwell} and \eqref{eq:regular_Maxwell} simplify to
\begin{align}
    f_l=f\UP[0]+\Kn \mLd \mD f_{l-1}, \label{eq:linear_origin_Maxwell} \\ 
    \hf_l=f\UP[0]+\Kn \hmLd \hmD \hf_{l-1}, \label{eq:linear_moment_Maxwell} \\
    \wf_l=f\UP[0]+\Kn \hmLd \hmD \wf_{l-1}. \label{eq:linear_moment_Maxwell}
\end{align}
    
We now proceed with the detailed proof of Theorem \ref{thm:linear}.
\begin{proof}[Proof of Theorem \ref{thm:linear}]
    We define the following auxiliary quantities:
    \begin{equation}
        \label{eq:def_mF_new}
        \mF\UP[k]_l=\mP_{\bbV\UP[k]} \mD f_l,\qquad \hmF\UP[k]_l=\mP_{\bbV\UP[k]} \hmD \hf_l.\qquad \wmF\UP[k]_l=\mP_{\bbV\UP[k]} \wmD \wf_l.
    \end{equation}
    Our objective is to demonstrate that $\mF\UP[0]_{2n}=\hmF\UP[0]_{2n}$ and $\mF\UP[0]_{2n-1}=\wmF\UP[0]_{2n-1}+O(\Kn^{2n})$. We first establish $\mF\UP[0]_{2n} = \hmF\UP[0]_{2n}$, and the second statement can be verified using a similar approach.  
    
    \paragraph{Step 1}
    This step derives the expressions of $\mF_{n}\UP[n]$ and $\hmF_n\UP[n]$. For any $l\leqslant n$, since $f_l\in\bbV\UP[l]$, the Maxwellian iteration can be reformulated as follows:
    \begin{equation}
        \label{eq:re_linear_Max}
        \begin{split}
            f_l&=f\UP[0]+\Kn \mP_{\bbV\UP[l]}\mLd \mD\mP_{\bbV\UP[l-1]} f_{l-1}, \\
            \hf_l&=f\UP[0]+\Kn \mP_{\bbV\UP[l]}\hmLd \hmD\mP_{\bbV\UP[l-1]} \hf_{l-1}
            =f\UP[0]+\Kn \mP_{\bbV\UP[l]}\mLd \mD\mP_{\bbV\UP[l-1]} \hf_{l-1}. \quad\text{(Lemma \ref{lemma:DL_deter})}
        \end{split}
    \end{equation}
    With the initial condition $f_0 = \hf_0 = f\UP[0]$, an inductive argument confirms that $f_n = \hf_n$.
    Consequently, it follows that:
    \begin{equation}
        \hmF_n\UP[n]=\mP_{\bbV\UP[n]}\hmD \hf_n=\mP_{\bbV\UP[n]}\mD\mP_{\bbV\UP[n]} f_n=\mP_{\bbV\UP[n]}\mD f_n=\mF_n\UP[n].
    \end{equation}
    
    \paragraph{Step 2}
    With \eqref{eq:linear_origin_Maxwell}, we have
    \begin{equation}
        \label{eq:rec_mF_1}
        \begin{split}
            \mF\UP[k]_l&=\mP_{\bbV\UP[k]}\mD f\UP[0]+\Kn\mP_{\bbV\UP[k]} \mD \mLd \mD f_{l-1} \\
            &= \mP_{\bbV\UP[k]}\mD f\UP[0]+\Kn\mP_{\bbV\UP[k]} \mD \mLd\left(\mP_{\bbV\UP[k+1]} \mD f_{l-1}+\mP_{\left(\bbV\UP[k+1]\right)^{\perp}} \mD f_{l-1}\right).
        \end{split}     
    \end{equation}
    By Lemma \ref{lemma:DL_hess}, the term
    $\mP_{\bbV\UP[k]} \mD \mLd \mP_{\left(\bbV\UP[k+1]\right)^{\perp}} \mD f_{l-1}$ vanishes. This simplifies \eqref{eq:rec_mF_1} to
    \begin{equation}
        \label{eq:rec_mF_2}
        \mF\UP[k]_l=\mP_{\bbV\UP[k]}\mD f\UP[0]+\Kn\mP_{\bbV\UP[k]} \mD \mLd \mP_{\bbV\UP[k+1]}\mF\UP[k+1]_{l-1}.
    \end{equation}
    Similarly, we have 
    \begin{equation}
        \label{eq:rec_hmF_2}
        \begin{split}
            \hmF\UP[k]_l&=\mP_{\bbV\UP[k]}\hmD f\UP[0]+\Kn\mP_{\bbV\UP[k]} \hmD \hmLd \mP_{\bbV\UP[k+1]}\hmF\UP[k+1]_{l-1} \\
            &=\mP_{\bbV\UP[k]}\mD f\UP[0]+\Kn\mP_{\bbV\UP[k]} \mD \mLd \mP_{\bbV\UP[k+1]}\hmF\UP[k+1]_{l-1}. \quad\text{(Using Lemma \ref{lemma:DL_deter}.)}
        \end{split}
    \end{equation}
    To establish $\hmF\UP[0]{2n} = \mF\UP[0]{2n}$, we initiate the following inductive process:
    \begin{itemize}
        \item Setting $l = n + 1$ and $k = n - 1$, and using $\hmF_n\UP[n] = \mF_n\UP[n]$, a direct comparison between \eqref{eq:rec_mF_2} and \eqref{eq:rec_hmF_2} confirms $\hmF_{n+1}\UP[n-1] = \mF_{n+1}\UP[n-1]$.
        \item Repeating this process with $k = n - 2$ and $l = n + 2$ leads to $\hmF_{n+2}\UP[n-2] = \mF_{n+2}\UP[n-2]$.
        \item Continuing this step-by-step down to $k = 0$ and $l = 2n$, we ultimately establish $\hmF\UP[0]{2n} = \mF\UP[0]{2n}$.
    \end{itemize}

    The analysis for the regularized system \eqref{eq:oper_regularized} follows a similar strategy. In Step 1, when obtaining $f_n$ with \eqref{eq:re_linear_Max}, we note that the operator the domain of operator $\wmD$ belongs to $\bbV\UP[n-1]$. Property \ref{prop:wmD} ensures $\wf_n = \hf_n = f_n$. It follows that:
    \begin{equation}
        \label{eq:wmF_n}
        \begin{split}
        \wmF_{n}\UP[n-1] &= \mP_{\bbV\UP[n-1]} \wmD \wf_n 
        = \mP_{\bbV\UP[n-1]} \hmD \hf_n + O(\Kn^{2n}) \\
        &= \mP_{\bbV\UP[n-1]} \hmD \mP_{\bbV\UP[n]} f_n + O(\Kn^{2n}) = \mF_n\UP[n-1]+O(\Kn^{2n}).
        \end{split}
    \end{equation}
    Following the same recurrence approach in Step 2, we establish that 
    \begin{equation}
        \label{eq:rec_wmF_2}
        \begin{split}
            \wmF_l\UP[k] &= \mP_{\bbV\UP[k]}\wmD f\UP[0]+\Kn\mP_{\bbV\UP[k]} \wmD \wmLd \mP_{\bbV\UP[k+1]}\wmF\UP[k+1]_{l-1} \\
            &= \mP_{\bbV\UP[k]}\hmD f\UP[0]+\Kn\mP_{\bbV\UP[k]} \hmD \hmLd \mP_{\bbV\UP[k+1]}\wmF\UP[k+1]_{l-1} + O(\Kn^{2n}) \\
            &= \mP_{\bbV\UP[k]}\hmD f\UP[0]+\Kn\mP_{\bbV\UP[k]} \hmD \hmLd \mP_{\bbV\UP[k+1]}\wmF\UP[k+1]_{l-1} + O(\Kn^{2n}),\quad k\leqslant n-1.
        \end{split}
    \end{equation}
    By successively setting $k = n - 1, n - 2, \ldots, 0$ and $l = 2n - 1 - k$, we ultimately prove $\wmF_{2n-1}\UP[0]=\mF_{2n-1}\UP[0]+O(\Kn^{2n})$.
\end{proof}


\subsubsection{Proof for the reduced model with quadratic collision operators}
\label{sec:full_order}

\begin{proof}[Proof of Theorem \ref{thm:full}]
    We aim to verify $\mP_{\bbV\UP[0]} \hmD \hf_{n+1}=\mP_{\bbV\UP[0]}\mD f_{n+1}$, and $\mP_{\bbV\UP[0]} \wmD \wf_{n+1}=\mP_{\bbV\UP[0]}\mD f_{n+1} + O(\Kn^{n+2})$. 

    \paragraph{Step 1} We first establish that $\wf_l=\hf_l=f_l$ for $l=0, 1, \cdots, n$ by induction. The base case $l=0$ naturally holds. Assume it holds for $0, 1, \cdots, l-1$, and consider the case of $l$. Similar to the proof of Theorem \ref{thm:linear}, we can derive 
    \begin{equation}
        \label{eq:eq_LDf}
        \mP_{\bbV\UP[l]} \mLd \mD f_{l-1}=\mP_{\bbV\UP[l]} \hmLd \hmD \hf_{l-1}.
    \end{equation}
    Moreover, by combining the definitions in \eqref{eq:def_hmQs_l}, \eqref{eq:def_wmQs_l}, and leveraging Lemmas \ref{lemma:hmQs} and \ref{lemma:Pn_mQs}, we have
    \begin{equation}
        \label{eq:eq_LQ}
        \begin{split}
            \mP_{\bbV\UP[l]}\hmLd \hmQs_l 
            &= \mP_{\bbV\UP[l]}\hmLd\mP_{\bbV\UP[n]}\mQs_l=\mP_{\bbV\UP[l]}\mLd\mQs_l. \\
            \mP_{\bbV\UP[l]}\wmLd \wmQs_l &= \mP_{\bbV\UP[l]}\hmLd\mP_{\bbV\UP[n]}\mQs_l=\mP_{\bbV\UP[l]}\mLd\mQs_l.
        \end{split}
    \end{equation}
    Substituting these into the Maxwellian iterations \eqref{eq:origin_Maxwell}, \eqref{eq:moment_Maxwell}, and \eqref{eq:regular_Maxwell} leads to
    \begin{equation}
        \label{eq:eq_fl}
        \begin{split} 
        \hf_l &= f\UP[0] + \mP_{\bbV\UP[l]}\hmLd \left( \Kn \hmD \hf_{l-1} - \hmQs_l \right) = f\UP[0] + \mP_{\bbV\UP[l]}\mLd \left( \Kn \mD f_{l-1} - \mQs_l \right)=f_l. \\
        \wf_l &= f\UP[0] + \mP_{\bbV\UP[l]}\wmLd \left( \Kn \wmD \wf_{l-1} - \wmQs_l \right) 
        = f\UP[0] + \Kn \mP_{\bbV\UP[l]}\hmLd\mP_{\bbV\UP[n]}\wmD f_{l-1} -  \mP_{\bbV\UP[l]}\wmLd\mQs_l \\ 
        &= f\UP[0] + \Kn \mP_{\bbV\UP[l]}\hmLd\hmD f_{l-1} -  \mP_{\bbV\UP[l]}\wmLd\mQs_l = f_l,
        \end{split}
    \end{equation}
    which completes the induction.

    \paragraph{Step 2}
    Following the same approach as in the derivations of \eqref{eq:rec_mF_2} and \eqref{eq:rec_hmF_2}, we obtain
    \begin{equation}
        \label{eq:exp_fnp1}
        \begin{split}
            \mP_{\bbV\UP[0]}\mD f_{n+1}&=\mP_{\bbV\UP[0]}\mD f\UP[0] +\mP_{\bbV\UP[0]}\mD\mLd\mP_{\bbV\UP[1]}\left(\Kn\mD f_n-\mQs_{n+1}\right), \\ 
            \mP_{\bbV\UP[0]}\hmD \hf_{n+1}&=\mP_{\bbV\UP[0]}\hmD f\UP[0] +\mP_{\bbV\UP[0]}\hmD\hmLd\mP_{\bbV\UP[1]}\left(\Kn\hmD \hf_n-\hmQs_{n+1}\right), \\
            \mP_{\bbV\UP[0]}\wmD \wf_{n+1}&=\mP_{\bbV\UP[0]}\wmD f\UP[0] +\mP_{\bbV\UP[0]}\wmD\wmLd\mP_{\bbV\UP[1]}\left(\Kn\wmD \wf_n-\wmQs_{n+1}\right), 
        \end{split}
    \end{equation}
    and
    \begin{equation}
        \label{eq:eq_DLfn}
        \mP_{\bbV\UP[0]}\mD\mLd\mP_{\bbV\UP[1]}\mD f_n=\mP_{\bbV\UP[0]}\hmD\hmLd\mP_{\bbV\UP[1]}\hmD \hf_n,\quad \mP_{\bbV\UP[0]}\mD f\UP[0]=\mP_{\bbV\UP[0]}\hmD f\UP[0].
    \end{equation}
    Besides, Lemma \ref{lemma:hmQs} indicates $\mP_{\bbV\UP[1]}\hmQs_{n+1}=\mP_{\bbV\UP[1]}\mQs_{n+1}$, ensuring $\mP_{\bbV\UP[0]} \mD f_{n+1}=\mP_{\bbV\UP[0]}\hmD \hf_{n+1}$. 

    For the regularized reduced system, it holds 
    \begin{equation}
        \label{eq:eq_wDLfn}
        \mP_{\bbV\UP[0]}\wmD\wmLd\mP_{\bbV\UP[1]}\wmD \wf_n=\mP_{\bbV\UP[0]}\hmD\hmLd\mP_{\bbV\UP[1]}\hmD \hf_n + O(\Kn^{2n}),\quad \mP_{\bbV\UP[0]}\wmD f\UP[0]=\mP_{\bbV\UP[0]}\hmD f\UP[0] + O(\Kn^{2n}).
    \end{equation}
    This result leads to $\mP_{\bbV\UP[0]} \wmD \wf_{n+1}=\mP_{\bbV\UP[0]}\hmD \hf_{n+1} + O(\Kn^{2n}) = \mP_{\bbV\UP[0]}\hmD \hf_{n+1}+O(\Kn^{n+2})$ (using $n\geqslant 2$), thereby completing the proof.  
\end{proof}
 


\begin{remark}
    Comparing Theorem \ref{thm:linear} and Theorem \ref{thm:full}, we observe that the linearized part achieves the accuracy order of $O(\Kn^{2n})$ or $O(\Kn^{2n-1})$, whereas the quadratic part only yields $O(\Kn^{n+1})$. When the impact of quadratic terms is minimal, the model can potentially provide a better approximation.
\end{remark}

\section{Conclusions}
\label{sec:conclusion}
This work presents a general framework for constructing high-order approximation models for the full Boltzmann equation. The key approach extends the methodology of \cite{LinR13} by decomposing the function space into subspaces of different orders with respect to $\Kn$, while introducing a novel and straightforward implementation of this decomposition. As demonstrated in Theorem \ref{thm:full}, our framework shows that a reduced system retaining terms up to $O(\Kn^n)$ can achieve $(n+1)$th-order accuracy. In particular, when the collision operator is linearized, Theorem \ref{thm:linear} reveals that even higher order of accuracy $O(\Kn^{2n})$ can be attained. These theoretical results establish a systematic procedure for deriving reduced models with arbitrary orders of accuracy.

To demonstrate practical applications, we have derived the 13-moment system with Burnett order accuracy and provided the general structure for super-Burnett order systems. Compared to conventional Burnett and super-Burnett equations, our new systems maintain more concise forms containing only first- and second-order derivatives. However, similar to Grad's moment equations, these new models may encounter challenges with hyperbolicity loss. Addressing this limitation, along with the development of appropriate boundary conditions, will be the focus of our future research.

\addcontentsline{toc}{section}{References}
\bibliographystyle{siamplain}
\bibliography{article}
\end{document}